\documentclass[12pt,american]{article}
\usepackage{comment}
\usepackage{lmodern}
\usepackage{lmodern}
\usepackage[T1]{fontenc}
\usepackage[latin9]{inputenc}
\usepackage{geometry}
\usepackage[active]{srcltx}
\usepackage{babel}
\usepackage{amsmath}
\usepackage{amsthm}
\usepackage{amssymb}
\usepackage{setspace}
\usepackage[authoryear]{natbib}
\usepackage{booktabs} 
\usepackage{caption}
\usepackage{subcaption}
\usepackage{array}
\usepackage{makecell}
\usepackage[unicode=true,pdfusetitle,
 bookmarks=true,bookmarksnumbered=false,bookmarksopen=false,
 breaklinks=false,pdfborder={0 0 1},backref=false,colorlinks=false]
 {hyperref}
\hypersetup{
 colorlinks,linkcolor={blue!70!black},citecolor={blue!50!black},urlcolor={blue!60!black}}

\makeatletter
\theoremstyle{plain}
\newtheorem{assumption}{\protect\assumptionname}
\theoremstyle{plain}
\newtheorem{thm}{\protect\theoremname}
\theoremstyle{plain}
\newtheorem{Remark}{Remark}
\theoremstyle{remark}

\usepackage{babel}
\usepackage{float}
\usepackage{mathrsfs}
\usepackage{breakurl}
\usepackage{bbm}
\usepackage{tikz}
\usepackage{centernot}
\usepackage{enumitem}
\allowdisplaybreaks

\newcommand{\customlabel}[2]{%
   \protected@write \@auxout {}{\string \newlabel {#1}{{#2}{\thepage}{#2}{#1}{}} }%
   \hypertarget{#1}{}
}

\providecommand{\remarkname}{Remark}

\providecommand{\assumptionname}{Assumption}

\providecommand{\theoremname}{Theorem}

\makeatother

\providecommand{\assumptionname}{Assumption}
\providecommand{\theoremname}{Theorem}

\begin{document}
\global\long\def\a{\alpha}%
\global\long\def\b{\beta}%
\global\long\def\g{\gamma}%
\global\long\def\d{\delta}%
\global\long\def\e{\epsilon}%
\global\long\def\l{\lambda}%
\global\long\def\t{\theta}%
\global\long\def\o{\omega}%
\global\long\def\s{\sigma}%
\global\long\def\G{\Gamma}%
\global\long\def\D{\Delta}%
\global\long\def\L{\Lambda}%
\global\long\def\T{\Theta}%
\global\long\def\O{\Omega}%
\global\long\def\R{\mathbb{R}}%
\global\long\def\N{\mathbb{N}}%
\global\long\def\Q{\mathbb{Q}}%
\global\long\def\I{\mathbb{I}}%
\global\long\def\P{\mathbb{P}}%
\global\long\def\E{\mathbb{E}}%
\global\long\def\B{\mathbb{\mathbb{B}}}%
\global\long\def\S{\mathbb{\mathbb{S}}}%
\global\long\def\V{\mathbb{\mathbb{V}}\text{ar}}%
\global\long\def\GG{\mathbb{G}}%
\global\long\def\TT{\mathbb{T}}%
\global\long\def\X{{\bf X}}%
\global\long\def\cX{\mathscr{X}}%
\global\long\def\cY{\mathscr{Y}}%
\global\long\def\cA{\mathscr{A}}%
\global\long\def\cB{\mathscr{B}}%
\global\long\def\cF{\mathscr{F}}%
\global\long\def\cM{\mathscr{M}}%
\global\long\def\cN{\mathcal{N}}%
\global\long\def\cG{\mathcal{G}}%
\global\long\def\cC{\mathcal{C}}%
\global\long\def\sp{\,}%
\global\long\def\es{\emptyset}%
\global\long\def\mc#1{\mathscr{#1}}%
\global\long\def\ind{\mathbf{\mathbbm1}}%
\global\long\def\indep{\perp}%
\global\long\def\any{\forall}%
\global\long\def\ex{\exists}%
\global\long\def\p{\partial}%
\global\long\def\cd{\cdot}%
\global\long\def\Dif{\nabla}%
\global\long\def\imp{\Rightarrow}%
\global\long\def\iff{\Leftrightarrow}%
\global\long\def\up{\uparrow}%
\global\long\def\down{\downarrow}%
\global\long\def\arrow{\rightarrow}%
\global\long\def\rlarrow{\leftrightarrow}%
\global\long\def\lrarrow{\leftrightarrow}%
\global\long\def\abs#1{\left|#1\right|}%
\global\long\def\norm#1{\left\Vert #1\right\Vert }%
\global\long\def\rest#1{\left.#1\right|}%
\global\long\def\bracket#1#2{\left\langle #1\middle\vert#2\right\rangle }%
\global\long\def\sandvich#1#2#3{\left\langle #1\middle\vert#2\middle\vert#3\right\rangle }%
\global\long\def\turd#1{\frac{#1}{3}}%
\global\long\def\ellipsis{\textellipsis}%
\global\long\def\sand#1{\left\lceil #1\right\vert }%
\global\long\def\wich#1{\left\vert #1\right\rfloor }%
\global\long\def\sandwich#1#2#3{\left\lceil #1\middle\vert#2\middle\vert#3\right\rfloor }%
\global\long\def\abs#1{\left|#1\right|}%
\global\long\def\norm#1{\left\Vert #1\right\Vert }%
\global\long\def\rest#1{\left.#1\right|}%
\global\long\def\inprod#1{\left\langle #1\right\rangle }%
\global\long\def\ol#1{\overline{#1}}%
\global\long\def\ul#1{\underline{#1}}%
\global\long\def\td#1{\tilde{#1}}%
\global\long\def\bs#1{\boldsymbol{#1}}%
\global\long\def\upto{\nearrow}%
\global\long\def\downto{\searrow}%
\global\long\def\pto{\overset{p}{\longrightarrow}}%
\global\long\def\dto{\overset{d}{\longrightarrow}}%
\global\long\def\asto{\overset{a.s.}{\longrightarrow}}%
\setlength{\abovedisplayskip}{6pt} \setlength{\belowdisplayskip}{6pt} 
\title{Inference on Welfare and Value Functionals \\
under Optimal Treatment Assignment\thanks{We thank Xu Cheng, Yanqin Fan, Sukjin Han, Patrick Kline, Soonwoo Kwon, Oliver Linton, Liyang Sun, Petra Todd, and conference participants at 2025 World Congress of the Econometric Society, 2025 California Econometrics Conference, 2025 Cowles Conference on Econometrics Celebrating Don Andrews for helpful comments and suggestions.}}
\author{Xiaohong Chen\thanks{Department of Economics and Cowles Foundation for Research in Economics,
Yale University, 
USA, xiaohong.chen@yale.edu. Chen thanks Cowles Foundation for research support.}, Zhenxiao Chen\thanks{Department of Economics, University of Pennsylvania,
USA, zxchen@upenn.edu.}, and Wayne Yuan Gao\thanks{Department of Economics, University of Pennsylvania, 
USA, waynegao@upenn.edu.}\\
 \textbf{~}}
\maketitle
\begin{abstract}
\noindent We provide theoretical results for the estimation and inference of a class of welfare and value functionals of the nonparametric conditional average treatment effect (CATE) function under optimal treatment assignment, i.e., treatment is assigned to an observed type if and only if its CATE is nonnegative. For the optimal welfare functional defined as the average value of CATE on the subpopulation with nonnegative CATE, we establish the $\sqrt{n}$ asymptotic normality of the semiparametric plug-in estimators and provide an analytical asymptotic variance formula. For more general value functionals, we show that the plug-in estimators are typically asymptotically normal at the 1-dimensional nonparametric estimation rate, and we provide a consistent variance estimator based on the sieve Riesz representer, as well as a proposed computational procedure for numerical integration on submanifolds. The key reason underlying the different convergence rates for the welfare functional versus the general value functional lies in that, on the boundary subpopulation for whom CATE is zero, the integrand vanishes for the welfare functional but does not for general value functionals. We demonstrate in Monte Carlo simulations the good finite-sample performance of our estimation and inference procedures, and conduct an empirical application of our methods on the effectiveness of job training programs on earnings using the JTPA data set. \\
 \textbf{~}
\end{abstract}

\section{\label{sec:Intro}Introduction }

In this paper, we study the estimation and inference on welfare and
value functionals of a given treatment under optimal (``first-best'')
treatment assignment.

Let $D_{i}\in\left\{ 0,1\right\} $ denote a certain binary treatment
for subject $i$, $\left(Y_{i}\left(0\right),Y_{i}\left(1\right)\right)$
denote the potential outcomes of interest, and $Y_{i}:=Y_{i}\left(D_{i}\right)\in \R$
denote the observed outcome. Let $X_{i}\in \R^d$ denote subject $i$'s observable
characteristics of $i$, which is distributed with density $f_0$. We suppose that researchers have access to
a random sample of training data $\left\{(D_{i},Y_{i},X_{i})\right\}_{i=1}^{n}$. 

Under the standard conditional unconfoundedness assumption $\rest{\left(Y_{i}\left(0\right),Y_{i}\left(1\right)\right)\indep D_{i}}X_{i}$
and the overlap condition $p_0\left(x\right):=\E\left[\rest{D_{i}}X_{i}=x\right]\in\left(0,1\right)$,
the conditional average treatment effect (CATE) defined by
\[
\text{CATE}\left(x\right):=\E\left[\rest{Y_{i}\left(1\right)-Y_{i}\left(0\right)}X_{i}=x\right]
\]
is identified from data by
\begin{align}
\text{CATE}\left(x\right)\equiv h_{0}\left(x\right) & :=\mu_{0}\left(x,1\right)-\mu_{0}\left(x,0\right),\label{eq:h0}
\end{align}
where $h_0:\R^d \mapsto \R$, and
\begin{equation}
\mu_{0}\left(x,d\right):=\E\left[\rest{Y_{i}}X_{i}=x,D_{i}=d\right]\label{eq:mu_0}
\end{equation}
is the nonparametric regression function of the outcome $Y_{i}$ on
$X_{i}$ and $D_{i}$. We maintain the conditional unconfoundedness
assumption and the overlap condition, and will thereafter simply refer
to $h_{0}$ as the CATE function. We further assume that $h_{0}$ belongs to a Holder class of functions with smoothness $s>1$.

We consider a standard scenario where policymakers can assign treatments
based on covariates, and focus on the following two core types of
welfare and value parameters. The first type is the maximized welfare
of the target population under optimal treatment assignment:
\begin{align}
W\left(h_{0}\right) & :=\int\left[h_{0}\left(x\right)\right]_{+}f\left(x\right)dx\label{eq:W_h0}
\end{align}
where $f$ is the marginal density of $x$ in the target population
and $\left[t\right]_{+}:=\max\left(t,0\right)$ is the rectified linear
unit (ReLU) function. $W\left(h_{0}\right)$ averages the CATE over
the population under the ``first-best'' treatment assignment rule:
``treat type $x$ if and only if $\text{CATE}\left(x\right)\geq0$,''
and is thus often referred to as the welfare under optimal treatment
assignment. We will thereafter refer to $W\left(h_{0}\right)$ as
the welfare functional.

Alternatively, one may also be interested in evaluating the average
of a value other than CATE over the population under optimal treatment
assignment:
\begin{align}
V\left(h_{0}\right) & :=\int\ind\left\{ h_{0}\left(x\right)\geq0\right\} v_0\left(x\right)f\left(x\right)dx\label{eq:V_h0}
\end{align}
where $v_0:\R^d \mapsto \R$ is a user-defined function that may be a utility function,
a cost function, or any economically meaningful function of the observed
covariate $x$. For example, setting $v_0\left(x\right)=a'x$ endows $V\left(h_{0}\right)$
with the interpretation as certain aggregate characteristics of the treated
population under optimal treatment assignment, and setting $v_0\equiv1$
implies that $V\left(h_{0}\right)\equiv\P_{f}\left(h_{0}\left(X_{i}\right)\geq0\right)$
becomes the share of the target population to be treated (i.e. with
nonnegative CATE).

In the formulation of $W\left(h_{0}\right)$ and $V\left(h_{0}\right)$
above, we take the density $f\left(x\right)$ to be known. This is
in itself relevant in settings where the covariate density of the
target population is configured or known/estimated from other sources
than the training sample used to estimated CATE. For example, CATE
may be estimated from a smaller pilot program, while the policymakers
are contemplating to implement the policy on a statewide or nation-wide
basis with a much larger population. That said, in this paper we also
consider an important case where $f$ is unknown and set to be the
covariate density in the underlying population of the training sample.
In this case, $f$ does not need to be estimated, as the integral
with respect to $f$ can be naturally approximated via sample average
in the training sample. This corresponds more closely to the ``empirical
welfare'' as considered in \citet*{kitagawa2018should}. We
provide results for this setting as well. 

~

In this paper, we establish inference results for the welfare and value functionals $W\left(h_{0}\right)$ and $V\left(h_{0}\right)$
with nonparametric estimated CATE $h_0$. The results can be summarized informally as follows.

For the welfare functional $W\left(h_{0}\right)$, we establish semiparametric
plug-in estimators of the welfare functional are asymptotically normal
at the parametric $\sqrt{n}$ rate and derive closed-form asymptotic
variance formulas, along with consistent asymptotic variance estimators.
The key insight of the $\sqrt{n}$ rate is geometric: by the definition
of the welfare functional, the integrand $h_{0}$$\left(x\right)$
vanishes on the boundary of the integration $\left\{ x:h_{0}\left(x\right)=0\right\} $,
neutralizing the non-smoothness of the indicator function $\ind\left\{ h_{0}\left(x\right)\geq0\right\} $. 

In contrast, for the value functional $V\left(h_{0}\right)$ with
general weight $v_{0}$ that does not vanish on the boundary $\left\{ x:h_{0}\left(x\right)=0\right\} $,
we show that the rate of convergence is slower than $\sqrt{n}$, and
is instead given by the 1-dimensional nonparametric regression rate
$n^{-\frac{s}{2s+1}}$ under appropriate conditions. We establish
asymptotic normality of the semiparametric plug-in estimator under
this irregular convergence rate, and provide a consistent variance
estimator based on the sieve Riesz representer. In particular, the
consistent variance estimator features a Hausdorff integral on the
boundary submanifold $\left\{ x:h_{0}\left(x\right)=0\right\} $,
for which we provide a numerical integration and differentiation procedure for the computation of submanifold integrals. 

We conduct an array of Monte Carlo experiments to document the good
finite-sample accuracy of our theoretical inferential results. We
show that the proposed standard error estimators perform well in finite
sample, and the corresponding confidence intervals based on the asymptotic
normality result and the standard error estimators have coverage probabilities
close to their nominal levels. These findings hold not only for the
$\sqrt{n}$-estimable welfare functional, but also for the value functional,
which is estimated at slower-than-$\sqrt{n}$ rate with standard errors
computed through numerical integration and differentiation.

We also apply our results to empirical data from the Job Training
Partnership Act (JTPA) data set. Following \citet*{kitagawa2018should},
we take 30-month post-program earning as the outcome variable and
consider two covariates: pre-program earning and education. We then
provide empirical estimates and confidence intervals for two parameters:
the welfare under first-best treatment assignment, which is $\sqrt{n}$
estimable, and the share of population to be treated under first-best
treatment assignment, which is not $\sqrt{n}$-estimable. As in \citet*{kitagawa2018should},
we also consider two different scenarios: one with the cost of the
treatment incorporated, and one without. These parameters have also
been estimated in \citet*{kitagawa2018should} under the label of
``nonparametric plug-in rule'' using kernel first-stages, but \citet*{kitagawa2018should}
only provide point estimates with no confidence intervals for them.
We use sieve (B-spline) first-stage nonparametric estimators and find
similar results to those in \citet*{kitagawa2018should}, and further
provide informative confidence intervals for both the welfare and
the share parameters. 

\subsubsection*{Closely Related Literature}

This work is a companion paper to \citet*{chen2025semiparametric},
and directly applies the general theoretical results there to handle
differentiation with respect to region of integration and semiparametric
estimation of integrals over submanifolds. Specifically, this work
focuses on the important context of treatment assignment
problems, and deal with two features that are not discussed
in \citet*{chen2025semiparametric}. First, CATE is defined as the
difference of two nonparametric regression functions between the treated and untreated
subpopulation (or the difference of two point evaluations of a nonparametric
function with treatment status defined as an argument too), and is
often estimated as the difference of two first-stage nonparametric
estimators. 
Second, in treatment assignment problems, researchers may face two
different scenario in terms of the distribution of the covariates:
sometimes, this distribution can be treated as known and may be different
from the covariate distribution in the experimental/observational
population from which CATE is estimated; in other times, one may want
to treat the covariate distribution as unknown and the same as the
experimental/observational population, and uses sample averages to
automatically incorporate the covariate distribution. In this paper,
we takes into account these special structures of the problem, and establish the inference results by providing lower-level sufficient conditions to the general theory in \citet*{chen2025semiparametric}.\footnote{Two concurrent papers by \cite*{cattaneo2025dist,cattaneo2025loc} also feature submanifold integrals, but focuses on 1-dimensional cases that arise from the specific context of boundary discontinuity designs. Our current paper also differs significantly from \cite*{cattaneo2025dist,cattaneo2025loc}, who focus on boundary discontinuity designs and boundary treatment effects, which are very different objects from the welfare and value functionals considered here under first-best treatment assignments. Consequently, the submanifolds (boundaries) in their settings are given or known based on the locations or a distance function, while in our current project the boundary submanifold is defined by the unknown and estimated CATE function.}

Our paper makes new contributions to the literature on
estimation and inference on a general value functional of a policy under optimal
treatment assignment. To the best of our knowleadge, all the existing work on limiting distributions of functionals of a policy under optimal treatment assignments considered $\sqrt{n}$-normality only. Our paper is the first to establish slower-than-root-n limiting distribution and inference results for irregular value functionals of a policy under optimal treatment assignment. Previously,  \citet*{bhattacharya2012inferring}
establishes $\sqrt{n}$-normality of the optimal welfare value
under budget constraint, using Nadaraya-Waston kernel estimator for
first-stage estimation of CATE. The important work of \citet*{kitagawa2018should}
focuses on a related but slightly different topic: empirical welfare
maximization within a constrained class of policy rules. That said,
\citet*{kitagawa2018should} also considers and reports empirical
estimates on the nonparametric plug-in rule, which can be interpreted
as a plug-in estimator of the first-best welfare, though there were
no theoretical results or confidence intervals for this estimate.
The recent papers by \citet*{park2024debiased}
and \citet*{whitehouse2025inference} use soft-max
functions to smooth the max function, and apply the debiased machine
learning approach to establish asymptotic normality and provide inferential
results, and \citet*{whitehouse2025inference} establishes $\sqrt{n}$ asymptotic normality of their soft-max welfare functional. \citet*{feng2024statistical}
analyzes binary treatment assignment under constraints and the asymptotic
property of the welfare of the policy under optimal cutoff choice, and establishes root-$n$ asymptotic normality of the functionals. Our paper complements these existing work in several ways.
First, we clarify that the $\sqrt{n}$-estimability of the welfare
functional is due to the fact that the integrand (CATE) by construction
vanishes on the boundary of the optimally treated population $\left\{ x:\text{CATE}\left(x\right)=0\right\} $,
which is specific to the welfare functional but generally not satisfied
for other types of value functionals. Second, we demonstrate that
the $\sqrt{n}$-estimability of the welfare functional can be attained without the use of smoothing/soft-max
functions.
Third and most importantly, our results extend well beyond the welfare functional, and
cover generic value functionals that may be slower than $\sqrt{n}$-estimable.

~

The rest of the paper is organized as follows. Section \ref{sec:RootN}
lays out the main model setup, and provides a conceptual explanation of why the welfare functional $W\left(h_{0}\right)$
can be $\sqrt{n}$-estimable while the value functional $V\left(h_{0}\right)$
is not in general. Section \ref{sec:WelfareFunc} then establishes
the inference results for the welfare functional $W\left(h_{0}\right)$,
while Section \ref{sec:ValueFunc} provides corresponding results
for the value function $V\left(h_{0}\right)$. We report numerical
results from Monte Carlo simulations in Section \ref{sec:Sim}, and conduct an empirical illustration in Section \ref{sec:Emp}. Proofs of theoretical results in the main text are available in Appendix \ref{app:Proof}.


\section{Model Setup and Functional Derivatives}\label{sec:RootN}

We first introduce the standard treatment effect model, the general welfare functional of interest and the maintained assumptions in Subsection \ref{sec:model}

\subsection{The Model and the Parameters of Interest}\label{sec:model}

We first state the maintained assumption in the paper. Let $\hat{h}\left(x\right):=\hat{\mu}\left(x,1\right)-\hat{\mu}\left(x,0\right)$ be a nonparametric estimator of the CATE function, in which
$\hat{\mu}\left(x,d\right)$ is a first-stage estimator of the
nonparametric regression model 
\begin{equation}
Y_{i}=\mu_{0}\left(X_{i},D_{i}\right)+\e_{i},\quad\E\left[\rest{\e_{i}}X_{i},D_{i}\right]=0,\quad\E\left[\rest{\e^2_{i}}X_{i},D_{i}\right]<\infty.\label{eq:NP_reg_mu}
\end{equation}
We impose the following basic assumptions in this paper.
\begin{assumption}[Model]\label{assu:main}  The training data and the model satisfy:
\begin{itemize}
\item[(a)] Training sample: the training data $\left\{(Y_{i},D_i,X_{i})\right\}_{i=1}^{n}$ is a random sample drawn from $(Y,D,X)\in \R\times \{0,1\}\times {\cal X}$ satisfying Model \eqref{eq:NP_reg_mu}, where  ${\cal X}$ is a bounded rectangular set (say $[0,1]^d$) in $\R^d$, and $X_i$ has its true unknown marginal density $f_0$ supported on ${\cal X}$.
\item[(b)] Overlap: $0< p_0(x):=\E[D_i|X_i = x]<1$.
\item[(c)] Smoothness of Regression Function: For $d\in\{0,1\}$, $\mu_{0}\left(\cdot,d\right)\in \Lambda^s({\cal X})$ with $s>1$.
\end{itemize}
\end{assumption}

We first provide a heuristic overview of our theoretical analysis,
and explain the key intuition why the welfare functional $W\left(h_{0}\right)$
may be $\sqrt{n}$-estimable (i.e. $W$ is a regular functionial)
while $V\left(h_{0}\right)$ is generally not (i.e., $V$ is a irregular
functional). 

We first introduce a more general value functional $\Phi \left(h_{0}\right)$ that nests $W\left(h_{0}\right)$ and  $V\left(h_{0}\right)$ as special cases:
\begin{equation}\label{Phi}
\Phi\left(h_{0}\right):=\int\ind\left\{ h_{0}\left(x\right)\geq0\right\} \phi\left(h_{0}(x),x\right)f\left(x\right)dx,
\end{equation}
where $\phi:\R\times \R^d \to\R$ is a known measurable mapping. We note that 
\begin{itemize}
\item $\Phi \left(h_{0}\right)=W\left(h_{0}\right)$ when 
$ \phi\left(h_{0}(x),x\right) =h_{0}(x)$;
\item $\Phi \left(h_{0}\right)=V\left(h_{0}\right)$ when $ \phi\left(h_{0}(x),x\right)=v_0(x)$ 
with $\partial_1 \phi\left(h_{0}(x),x\right)=0$. 
\end{itemize}

\begin{assumption}[Functional]\label{assu:main-functional}  The functional $\Phi$ satisfies
\begin{itemize}
\item[(a)] $\phi:\R\times \R^d \to\R$ is continuously differentiable with respect to its first argument.
\item[(b)] Target Density: The target density $f$ of $X$ is absolutely continuous w.r.t. $f_{0}$ with uniformly bounded Radon-Nikodym derivative $\l:=f/f_{0}$.
\item[(c)] Regular Level Set: $h_0 (\cdot):=\mu_{0}\left(\cdot,1\right)-\mu_{0}\left(\cdot,0\right)$ satisfies $\norm{\Dif_{x} h_{0}\left(x\right)}\geq\ul{c}>0$ on the level set $\left\{ x\in {\cal X}:h_{0}\left(x\right)=0\right\} $.
\end{itemize}
\end{assumption}

\subsection{Functional Derivatives via Generalized Leibniz rule}\label{sec:Diff}

By the standard semiparametric theory on the estimation of functionals
of nonparametric regression functions, the asymptotic property of
the semiparametric plug-in estimator $\Phi\left(\hat{h}\right)$
can be analyzed via the functional derivative
of $\Phi\text{\ensuremath{\left(h\right)}}$ w.r.t.
$h_0$ in the direction of $h-h_{0}$, i.e., writing $h_{t}:=h_{0}+t\left(h-h_{0}\right)$,
\begin{align}
D_{h}\Phi\left(h_{0}\right)\left[h-h_{0}\right] & :=\rest{\frac{d}{dt}\Phi\left(h_{t}\right)}_{t=0}=\rest{\frac{d}{dt}\int\ind\left\{ h_{t}\left(x\right)\geq0\right\} \phi\left(h_t(x),x\right)f\left(x\right)dx}_{t=0}.\label{eq:Dh_Phi}\\
D_{h}W\left(h_{0}\right)\left[h-h_{0}\right] & :=\rest{\frac{d}{dt}W\left(h_{t}\right)}_{t=0}=\rest{\frac{d}{dt}\int\left[h_{t}\left(x\right)\right]_{+}f\left(x\right)dx}_{t=0}\label{eq:Dh_Wh}\\
D_{h}V\left(h_{0}\right)\left[h-h_{0}\right] & :=\rest{\frac{d}{dt}V\left(h_{t}\right)}_{t=0}=\rest{\frac{d}{dt}\int\ind\left\{ h_{t}\left(x\right)\geq0\right\} v_0\left(x\right)f\left(x\right)dx}_{t=0}\label{eq:Dh_Vh}
\end{align}
The presence of the ReLU/max function $\left[t\right]_{+}$ in $D_{h}W\left(h_{0}\right)$ and the indicator function $\ind\left\{ t\geq0\right\} $ induces a point of nonsmoothness at $t=0$, where $\left[t\right]_{+}$ is nondifferentiable and $\ind\left\{ t\geq0\right\} $ is discontinuous. This complicates the calculation of the functional derivatives in \eqref{eq:Dh_Wh} and \eqref{eq:Dh_Vh}, though to different degrees.

We now provide an overview of the key difference between the welfare and value functions from the perspective of the generalized Leibniz rule,
which has the following generic form regarding the total time derivative of integrals
with changing integrand and changing region of integration:
\begin{equation}
\frac{d}{dt}\int_{\O_{t}}G_{t}\left(x\right)dx.\label{eq:time_deriv}
\end{equation}
The generalized Leibniz rule,\footnote{See, for example, Theorem 4.2 of \cite{delfour2001shapes}.} 
states that, under mild regularity conditions,
\begin{align}
\frac{d}{dt}\int_{\O_{t}}G_{t}\left(x\right)dx & =\underset{({\bf I})}{\underbrace{\int_{\O_{t}}\frac{\p}{\p t}G_{t}\left(x\right)dx}}+\underset{\left({\bf II}\right)}{\underbrace{\int_{\p\O_{t}}\left\langle {\bf n}_{t}\left(x\right),{\bf v}_{t}\left(x\right)\right\rangle G_{t}\left(x\right)dS_{\p\O_{t}}\left(x\right)}}\label{eq:Reynolds}
\end{align}
where term (I) captures the effect of the change in the integrand
$G_{t}\left(x\right)$ with the region of integration $\O_{t}$ held
fixed, while term (II) captures the effect of the change in the region
of integration $\O_{t}$ with the integrand $G_{t}\left(x\right)$
held fixed. The somewhat ``nonstandard'' term (II) warrants some
more explanations: $\p\O_{t}$ denotes the boundary of $\O_{t}$,
${\bf n}_{t}\left(x\right)$ is the outward-pointing unit normal vector,
${\bf v}_{t}\left(x\right)$ is the velocity vector associated with
the time movement in the $\p\O_{t}$, and $S_{\p\O_{t}}\left(x\right)$
denotes the surface measure on the boundary $\p\O_{t}$. Note that,
when $x$ is one-dimensional and $\O_{t}=\left[a_{t},b_{t}\right]$,
\eqref{eq:Reynolds} specializes to the standard Leibniz rule:
\[
\frac{d}{dt}\int_{a_{t}}^{b_{t}}G_{t}\left(x\right)dx=\int_{a_{t}}^{b_{t}}\frac{\p}{\p t}G_{t}\left(x\right)dx+G_{t}\left(b_{t}\right)\frac{d}{dt}b_{t}-G_{t}\left(a_{t}\right)\frac{d}{dt}a_{t}.
\]

We observe that $D_{h}\Phi\left(h_{0}\right)$, $D_{h}W\left(h_{0}\right)$ and $D_{h}V\left(h_{0}\right)$
are all of the form \eqref{eq:time_deriv}, with the same parametrized
region of integration 
\[
\O_{t}:=\left\{ x\in \R^d:h_{t}\left(x\right)\geq0\right\}~,~~\O_{0}:=\left\{ x\in \R^d:h_{0}\left(x\right)\geq0\right\}
\]
and $\p\O_{0}= \left\{ x\in \R^d:h_{0}\left(x\right)=0\right\} $ under mild regularity conditions on $h_0$.

Applying the generalized Leibniz rule \eqref{eq:Reynolds} to the functional $\Phi(h)$, we obtain:
\begin{align}
D_{h}\Phi\left(h_{0}\right)\left[h-h_{0}\right]= & \underset{({\bf I})}{\underbrace{\int_{\O_{0}} \partial_1 \phi\left(h_{0}(x),x\right)\left(h\left(x\right)-h_{0}\left(x\right)\right)f\left(x\right)dx}}\nonumber \\
 & +\underset{({\bf II})}{\underbrace{\int_{\p\O_{0}}\left\langle {\bf n}_{0}\left(x\right),{\bf v}_{0}\left(x\right)\right\rangle \phi\left(h_{0}(x),x\right)f\left(x\right)dS_{\p\O_{0}}\left(x\right)}}\label{eq:Dh_Phi}
\end{align}
where the first term (I) is a full-dimensional Lebesgue integral (in $ \R^d$),
while the second term (II) is a lower-dimensional boundary integral.
In particular, when $\phi\left(h_{0}(x),x\right) f(x)$ does not vanish on $\p\O_{0}=\left\{ x\in \R^d:h_{0}\left(x\right)=0\right\}$, the second term (II) cannot
be ignored, despite $\p\O_{0}$ has Lebesgue measure zero in $ \R^d$. In fact, under Assumption \ref{assu:main-functional}(c) (see, e.g., \citet*{chen2025semiparametric}), the second term (II) of \eqref{eq:Dh_Phi} can be expressed as 
\begin{align}
\int_{\p\O_{0}}\left\langle {\bf n}_{0}\left(x\right),{\bf v}_{0}\left(x\right)\right\rangle \phi\left(h_{0}(x),x\right)f\left(x\right)dS_{\p\O_{0}}\left(x\right)=  \int_{\p\O_{0} }\frac{h(x)-h_0 (x)}{\norm{\Dif_x h_{0}\left(x\right)}}\phi\left(h_{0}(x),x\right)f\left(x\right)d{\cal H}^{d-1}\left(x\right),\label{eq:Dh_Phi-II}
\end{align}
where ${\cal H}^{d-1}$ denotes the $(d-1)$ dimensional Hausdorff measure (see \citet*{chen2025semiparametric}), which coincides with the $(d-1)$ dimensional Lebesgue measure in $\R^{d-1}$. 
Thus the boundary integral term (II) of \eqref{eq:Dh_Phi} is a lower-dimensional integral functional that only extracts information about $h_{0}$ on a Lebesgue measure-$0$ set (in $\R^d$).


%

For the welfare function $\Phi(h_0)=W(h_0)$, plugging $\phi(h_0(x),x) = h_0(x)$ into \eqref{eq:Dh_Phi} yields 
\begin{equation}
D_{h}W\left(h_{0}\right)\left[h-h_{0}\right]=\int\ind\left\{ h_{0}\left(x\right)\geq0\right\} \left[h\left(x\right)-h_{0}\left(x\right)\right]f\left(x\right)dx,\label{eq:Dh_Wh_form}
\end{equation}
where the term (II) vanishes since $\phi(h_0(x),x) = h_0(x) = 0$ on the boundary $\p\O_{0}$, and the term (I) is a full-dimensional Lebesgue integral functional of $h-h_{0}$. 

For the value function $\Phi(h_0)=V(h_0)$, plugging $\phi(h_0(x),x) = v_0(x)$ with $\partial_1 \phi\left(h_{0}(x),x\right)=0$ into \eqref{eq:Dh_Phi} and \eqref{eq:Dh_Phi-II} yields 
\begin{align}
D_{h}V\left(h_{0}\right)\left[h-h_{0}\right]= & \int_{\left\{x\in\R^d: h_{0}\left(x\right)=0\right\} }\frac{\left(h\left(x\right)-h_0\left(x\right)\right)}{\norm{\Dif_{x} h_{0}\left(x\right)}}v_0\left(x\right)f\left(x\right)d{\cal H}^{d-1}\left(x\right).\label{eq:Dh_Vh_form}
\end{align}
which is not $0$ as long as $v_0(x)f(x)$ does not vanish
on the boundary $\p\O_{0}= \left\{ x\in \R^d:h_{0}\left(x\right)=0\right\}$. For example, setting $v_0\left(x\right)\equiv1$
yields $V\left(h_{0}\right)=P_{f}\left(h_{0}\left(X_{i}\right)\geq0\right)$,
the share of population with nonnegative CATE, and the boundary integral term (II) does not vanish. In fact, as long as $v_0(x)f(x)\neq 0$ on  the boundary $\p\O_{0}= \left\{ x\in \R^d:h_{0}\left(x\right)=0\right\}$, $D_{h}V\left(h_{0}\right)$ is a non-zero $(d-1)$ dimensional
integral functional that only extracts information about $h_{0}$
on a Lebesgue measure-$0$ set (in $\R^d$), akin to a point evaluation of a nonparametric estimation.

\subsection{Key Difference between the Welfare and Value Functionals}

Let $L^2(f)$ denote the Hilbert space of square integrable (against $f$) functions with the inner product $\left\langle g,h\right\rangle _{2,f}:=\int g\left(x\right)h\left(x\right)f\left(x\right)dx$. For the CATE function $h_0 \in L^2(f)$, it is well-known that a linear functional $L\left[h-h_{0}\right]$ is bounded (or equivalently, continuous) if and only if
\[
 \sup_{\nu \neq 0, \nu \in L^2 (f)}\frac{|L\left[\nu (\cdot)\right]|^2}{E_f[|\nu (X)|^2]}<\infty
\]
which is a necessary and sufficient condition for the existence of a Riesz representer $\nu^*\in L^2(f)$ such that 
\[
L[\nu]=\left\langle \nu^*,\nu\right\rangle _{2,f}~~~\text{for all}~\nu \in L^2(f)
\]
This in turn is a necessary condition for any plug-in estimator of the linear functional $L\left[\hat{h}-h_{0}\right]=\left\langle \nu^*,\hat{h}-h_{0}\right\rangle _{2,f}$ to converge to zero at a root-$n$ rate.

For the welfare functional, its linear directional derivative functional $L[\nu]=D_{h}W\left(h_{0}\right)\left[\nu\right]$ given in \eqref{eq:Dh_Wh_form}, we immediately see that $\nu^{*}\left(x\right):=\ind\left\{ h_{0}\left(x\right)\geq0\right\} $
is the Riesz representer of the linear functional $D_{h}W\left(h_{0}\right)\left[\nu\right]$
in the Hilbert space $L^2(f)$, and that this Riesz representer has bounded
norm
\[
\norm{\nu^{*}}^{2}:=\int\ind^{2}\left\{ h_{0}\left(x\right)\geq0\right\} f\left(x\right)dx\leq1,
\]
and thus, by well-known results in, say, \citet*{chen2014sieveIrregular}, \citet*{chen2014sieve} and \cite{chenpouzo2015sieve},
the linear functional $D_{h}W\left(h_{0}\right)\left[\nu\right]$ is a regular (i.e., $\sqrt{n}$-estimable) functional
under appropriate conditions.

In contrast, for the general value functional, the linear functional corresponding to its directional derivative
$D_{h}V\left(h_{0}\right)\left[\nu \right]$ given in \eqref{eq:Dh_Vh_form} 
does not have a well-defined Riesz representer in the Hilbert space $L^2(f)$.
It is well-known that, according to Lemma 3.3 of \cite{chenpouzo2015sieve}, 
Consequently, the functional $V$ becomes an irregular functional that
cannot be estimated at $\sqrt{n}$ rate.


The above provides an intuitive explanation of why the welfare functional
$W\left(h_{0}\right)$ is very special relative to general types of
value functionals $V\left(h_{0}\right)$ or $\Phi\left(h_{0}\right)$,
and clarifies why $W\left(h_{0}\right)$ could be $\sqrt{n}$-estimable
while others generally cannot. In subsequent sections, we provide
formal conditions and theorems that establish the $\sqrt{n}$-normality
of plug-in estimators of the welfare functional $W\left(h_{0}\right)$,
as well as the slower-than-$\sqrt{n}$ asymptotic normality for the
value functional $V\left(h_{0}\right)$. 

~

\begin{Remark}[An Alternative View]
We also briefly discuss alternative view of the determinant of $\sqrt{n}$-estimability
of the welfare fucntional $W\left(h_{0}\right)$, based on the Lipchitz
continuity of the ReLU/max function $\left[t\right]_{+}$. We note
that, under mild conditions ensuring that the level set $\left\{ x:h_{0}\left(x\right)=0\right\} $
has Lebesgue measure $0$, we may interchange the order of differentiation
and integral based on the almost sure differentiability of $\left[h_{t}\left(x\right)\right]_{+}$
and the dominant convergence theorem:
\begin{align}
D_{h}W\left(h_{0}\right)\left[h-h_{0}\right] & =\rest{\frac{d}{dt}\int\left[h_{t}\left(x\right)\right]_{+}f\left(x\right)dx}_{t=0}\nonumber \\
& =\int\rest{\frac{d}{dt}\left[h_{t}\left(x\right)\right]_{+}}_{t=0}f\left(x\right)dx\label{eq:Dw_Lip}\\
& =\int\ind\left\{ h_{0}\left(x\right)\geq0\right\} \left(h\left(x\right)-h_{0}\left(x\right)\right)f\left(x\right)dx,\nonumber 
\end{align}
which yields the same formula as in \eqref{eq:Dh_Wh_form}. However,
for general value functional $V\left(h_{0}\right)$, the indicator
function $\ind\left\{ t\geq0\right\} $ is no longer Lipchitz, and
there is no analog of \eqref{eq:Dw_Lip}: the functional derivative
$D_{h}V\left(h_{0}\right)$ needs to be derived using the generalized
Leibniz rule as described above (or its many variants or generalized
forms in differential geometry and geometric measure theory). 
\end{Remark}

\section{\label{sec:WelfareFunc}Estimation and Inference of the Welfare Functional }

In this section, we focus on the welfare functional defined in \eqref{eq:W_h0},
which can also be equivalently written as a functional of $\mu_{0}$
as follows:
\begin{align}
W\left(h_{0}\right) & :=\int\left[h_{0}\left(x\right)\right]_{+}f\left(x\right)dx\nonumber \\
\equiv\ol W\left(\mu_{0}\right) & :=\int\left[\mu_{0}\left(x,1\right)-\mu_{0}\left(x,0\right)\right]_{+}f\left(x\right)dx\label{eq:W_mu0-1}
\end{align}
We write out the two equivalent definitions of the functionals based
on $h_{0}$ and $\mu_{0}$, since each representation has its own
merit. The representations $W\left(h_{0}\right)$ based on the CATE
function $h_{0}$ is clearer in terms of its interpretation: the condition
$h_{0}\left(x\right)\geq0$ is a direct optimal treatment assignment,
and this representation has been adopted in previous work such as
\citet*{kitagawa2018should}. On the other hand, the representations
$\ol W\left(\mu_{0}\right)$ based on $\mu_{0}$ is clearer in terms
of the underlying nonparametric regression function, which is notationally
easier to work with in our subsequent semiparametric asymptotic analysis.

We consider two different setups for the estimation of $W\left(h_{0}\right)$,
depending on how we treat the marginal density $f\left(x\right)$. 

In the first setup, we treat $f$ as known and the functional $W\left(\cd\right)$
as a  known transformation of $h_{0}$. This is relevant
in cases where the covariate density $f\left(x\right)$ of the target
population is either configured or known/estimated from other sources
than the training sample used to estimated CATE. In the formulation
of $W\left(h_{0}\right)$ and $V\left(h_{0}\right)$ above, we take
the density $f\left(x\right)$ to be known. For example, CATE may
be estimated from a smaller pilot program (with sample size $n$),
while the policymakers are contemplating to implement the policy on
a statewide or nation-wide basis with a much larger target population,
whose covariate density may either be known at the population level
or estimated from an alternative data source with much larger sample
size $N>>n$ (so that the sampling uncertainty in the estimation of
$f$ becomes negligible relative to that in the estimation of $h_{0}$). 

In the second setup, we take $f$ to be unknown and set it to $f_{0}$,
the covariate density in the underlying population of the training
sample. In this case, $f$ does not need to be explicitly estimated,
but the integral in $W\left(\cd\right)$ with respect to $f$ can
be naturally approximated via sample average in the training sample.
This corresponds more closely to the ``empirical welfare'' as considered
in \citet*{kitagawa2018should}. We also provide results for this
setting as well.

\subsection{\label{subsec:Welfare_IntF}Welfare Functional Under Known Density
$f$}

We start with the first setup, where the covariate density $f$ is
taken to be known and $W\left(\cd\right)$ is treated as a deterministic
known integral of $h_{0}$ w.r.t. $f$. In this case, we can define a simple nonparametric plug-in
estimators of $W\left(h_{0}\right)$ as
\[
W\left(\hat{h}\right)\equiv\ol W\left(\hat{\mu}\right)=\int\left[\hat{h}\left(x\right)\right]_{+}f\left(x\right)dx.
\]
%

We first state some key assumptions. Let $\hat{h}\left(x\right):=\hat{\mu}\left(x,1\right)-\hat{\mu}\left(x,0\right)$ be a nonparametric estimator of the CATE function, in which
$\hat{\mu}\left(x,d\right)$ is a first-stage estimator of the
nonparametric regression model 
\begin{equation}
Y_{i}=\mu_{0}\left(X_{i},D_{i}\right)+\e_{i},\quad\E\left[\rest{\e_{i}}X_{i},D_{i}\right]=0,\quad\E\left[\rest{\e^2_{i}}X_{i},D_{i}\right]<\infty.\label{eq:NP_reg_mu}
\end{equation}

\begin{assumption}
\label{assu:FS_quart_rate}Fist-Stage Convergence: $\norm{\hat{\mu}-\mu_{0}}_{\infty}=o_{p}\left(n^{-1/4}\right)$. 

\end{assumption}

\begin{thm}[$\sqrt{n}$-asymptotic normality for the welfare functional]
\label{thm:Welfare_intF} Under Assumptions \ref{assu:main} and \ref{assu:main-functional}(b)(c), we have:
\[
\nu^{*}\left(x,d\right):=\ind\left\{ h_{0}\left(x\right)\geq0\right\} \l\left(x\right)\left(\frac{d}{p_{0}\left(x\right)}-\frac{1-d}{1-p_{0}\left(x\right)}\right)
\]
is the Riesz representer for the linear functional $D_{\mu}\ol W\left(\mu_{0}\right)\left[\cdot\right]$.\\
If furthermore Assumption \ref{assu:FS_quart_rate} holds, we have:
\[
\sqrt{n}\left(\ol W\left(\hat{\mu}\right)-\ol W\left(\mu_{0}\right)\right)=\frac{1}{\sqrt{n}}\sum_{i=1}^{n}\nu^{*}\left(X_{i},D_{i}\right)\e_{i}+o_{p}\left(1\right).
\]
%
Then:
\[
\sqrt{n}\left(W\left(\hat{h}\right)-W\left(h_{0}\right)\right)\equiv\sqrt{n}\left(\ol W\left(\hat{\mu}\right)-\ol W\left(\mu_{0}\right)\right)\dto\cN\left(0,\s_{W}^{2}\right),
\]
with
\begin{align*}
\s_{W}^{2} & :=\E\left[\left(\nu^{*}\left(X_{i},D_{i}\right)\e_{i}\right)^2\right]=\E\left[\frac{\ind\left\{ h_{0}\left(X_{i}\right)\geq0\right\} \l^{2}\left(X_{i}\right)\s_{\e}^{2}\left(X_{i}\right)}{p_{0}\left(X_{i}\right)\left(1-p_{0}\left(X_{i}\right)\right)}\right],
\end{align*}
where $\s_{\e}^{2}\left(x\right):=\E\left[\rest{\e_{i}^{2}}X_{i}=x\right]$. 
\end{thm}
Theorem \ref{thm:Welfare_intF} suggests the following natural estimator for the asymptotic variance $\s_{W}^{2}$:
\begin{align}\label{eq:se_analytic}
\hat{\s}_{W}^{2} & :=\frac{1}{N}\sum_{i}\frac{\ind\left\{ \hat{h}\left(X_{i}\right)\geq0\right\} \l^{2}\left(X_{i}\right)\hat{u}_{i}^{2}}{\hat{p}\left(X_{i}\right)\left(1-\hat{p}\left(X_{i}\right)\right)}
\end{align}
where $\hat{u}_{i}:=Y_{i}-\hat{h}\left(X_{i}\right)$. This requires nonparametric estimation of propensity score function $p(x)$, as well as the knowledge (or nonparametric estimation) of the density $f_0$ or the Radon-Nikodym derivative $\l(x)$.

Alternatively and more preferably, we can use the sieve-based asymptotic variance estimator, which does not require estimation or knowledge of $p(x)$ and $\l(x)$. To do so, we first clarify some subtlety in the definition of the sieve in the first-stage nonparametric estimation of CATE as $\hat{\mu}(x,1)-\hat{\mu}(x,0)$, where $\hat{\mu}$ is a linear
sieve estimator of $\mu_{0}\left(x,d\right)$ under random design
on $\left(X_{i},D_{i}\right)$ with $D_{i}$ being binary. 

In practice,
$\hat{\mu}\left(x,1\right)$ is estimated with $K_1$ linear series
in the treated subsample, while $\hat{\mu}\left(x,0\right)$ is estimated
separately with potentially different $K_0$ linear series in the untreated subsample.
However, since we treat $D_{i}$ as a random variable, we cannot directly
treat $\hat{\mu}\left(x,1\right)$ and $\hat{\mu}\left(x,0\right)$
as two completely separate nonparametric estimators with exogenously
given sample sizes. Instead, we treat $\hat{\mu}$ as the least square
estimator of 
\begin{equation} \label{Sieve_OLS}
Y_{i}=D_{i}\psi^{\left(K_1\right)}\left(X_{i}\right)^{'}\b_{1}+\left(1-D_{i}\right)\psi^{\left(K_0\right)}\left(X_{i}\right)^{'}\b_{0}+u_{i}    
\end{equation}
where  $\psi^{(K_1)}(x) = (\psi_1(x), \ldots, \psi_{K_1}(x))'$ and $\psi^{(K_0)}(x) = (\psi_{K_1+1}(x), \ldots, \psi_{K_1+K_0}(x))'$ denote the B-spline basis functions used to estimate $\mu_0(x, 1)$ and $\mu_0(x, 0)$, respectively, with sieve dimensions $K_1$ and $K_0$.

Define the vector-valued function $\overline{\psi}(x)$ such that its $k$th component equals $d\psi_k(x)$ for $k \in \{1, \ldots, K_1\}$ and $(1-d)\psi_k(x)$ for $k \in \{K_1+1, \ldots, K_1+K_0\}$. Following \citet*{chen2018optimal}, specifically equations (6) and (7), the variance (or standard error) estimator in our setting takes the form
\begin{equation}\label{eq:se_sieve}
   \hat{\s}_{W}^{2}:=D_\mu \overline{W}\left(\hat{\mu}\right)\left[\overline{\psi}\right]'\hat{\O}D_\mu \overline{W} \left(\hat{\mu}\right)\left[\overline{\psi}\right], 
\end{equation}
where $\hat{\Omega}$ denotes the estimated asymptotic covariance matrix of the OLS estimators in \eqref{Sieve_OLS} given by 
\[
\hat{\O}:=\left(\Psi^{\left(2K\right)}\Psi^{\left(2K\right)'}\right)^{-1}\left(\frac{1}{n}\sum_{i=1}^{n}u_{i}^{2}\Psi^{\left(2K\right)}\Psi^{\left(2K\right)'}\right)\left(\Psi^{\left(2K\right)}\Psi^{\left(2K\right)'}\right)^{-1}
\]
and the estimated directional derivative vector $D_\mu \ol{W}\left(\hat{\mu}\right)\left[\ol{\psi}\right]$ is given by
\[
D_\mu \ol{W}\left(\hat{\mu}\right)\left[\ol{\psi}\right] = 
\begin{pmatrix}
\int_{\{ \hat{\mu}(x,1) - \hat{\mu}(x,0)\geq 0 \}} \psi^{(K_1)}(x)f(x)dx\\
-\int_{\{  \hat{\mu}(x,1) - \hat{\mu}(x,0)\geq 0 \}} \psi^{(K_0)}(x) f(x)dx
\end{pmatrix},
\]
where the minus sign in front of the sieve terms in $\psi^{K_0}(x)$ is due to the presence of the minus sign in front of $\mu(x,0)$ in  $\text{CATE}(x)=\mu(x,1)-\mu(x,0)$. 

We use Sobol points to numerically compute the integral above. See the simulation section for more details.

\subsection{\label{subsec:Welfare_mean}Welfare Functional Under Unknown Density
$f=f_{0}$}

In certain cases, such as in \citet*{kitagawa2018should},
one might be interested in the welfare functional under the original
distribution of covariates $F_{0}$, which may not be known or controlled.
Given the random sample of $\left(Y_{i},D_{i},X_{i}\right)_{i=1}^{n}$
used to estimate the CATE $h_0$, it is natural to use the sample mean
$\frac{1}{n}\sum_{i=1}^{n}[\cdot]$ as an estimator of the population expectation
$\int [\cdot]dF_0$, in which case a natural plug-in estimator of $W\left(h_{0}\right)\equiv\ol W\left(\mu_{0}\right)$
is given by
\[
\hat{W}\left(\hat{h}\right)\equiv\hat{\ol W}\left(\hat{\mu}\right):=\frac{1}{n}\sum_{i=1}^{n}\left[\hat{\mu}\left(X_{i},1\right)-\hat{\mu}\left(X_{i},0\right)\right]_{+}\equiv\frac{1}{n}\sum_{i=1}^{n}\left[\hat{h}\left(X_{i}\right)\right]_{+}.
\]
Clearly, the approximation of the integral introduces an additional
source of randomness, but the result established in the last subsection
continues to be useful in this case.
\begin{thm}
\label{thm:Welfare_mean} Under Assumptions \ref{assu:main}, \ref{assu:main-functional}(b)(c) and \ref{assu:FS_quart_rate},
\begin{align*}
 & \sqrt{n}\left(\hat{W}\left(\hat{h}\right)-W\left(h_{0}\right)\right)\equiv\sqrt{n}\left(\hat{\ol W}\left(\hat{\mu}\right)-\ol W\left(\mu_{0}\right)\right)\\
= & \frac{1}{\sqrt{n}}\sum_{i=1}^{n}\left(\left[h_{0}\left(X_{i}\right)\right]_{+}-W\left(h_{0}\right)+\nu^{*}\left(X_{i},D_{i}\right)\e_{i}\right)+o_{p}\left(1\right)\dto\cN\left(0,\ol{\s}_{W}^{2}\right)
\end{align*}
where $\ol{\s}_{W}^{2}:=\E\left[\left(\left[h_{0}\left(X_{i}\right)\right]_{+}-W(h_0)+\nu^{*}\left(X_{i},D_{i}\right)\e_{i}\right)^2\right]=\text{\text{Var}\ensuremath{\left(\left[h_{0}\left(X_{i}\right)\right]_{+}\right)}}+\s_{W}^{2}$. 
\end{thm}

The standard errors can be computed similarly, based on straightforward adaptations of the analytical formula \eqref{eq:se_analytic} or the sieve-based formula \eqref{eq:se_sieve} in Section 3.1.

\section{\label{sec:ValueFunc}Estimation and Inference of the Value Functional}

\subsection{\label{subsec:ValueFunc_IntF}Value Functional Under Known
Density $f$}

We now turn to the general value functional given by 
\begin{align*}
V\left(h_{0}\right) & :=\int\ind\left\{ h_{0}\left(x\right)\geq0\right\} v_0\left(x\right)f\left(x\right)dx\\
\equiv\ol V\left(\mu_{0}\right) & :=\int\ind\left\{ \mu_{0}\left(x,1\right)-\mu_{0}\left(x,0\right)\geq0\right\} v_0\left(x\right)f\left(x\right)dx.
\end{align*}
The simple plug-in estimators are defined as
\begin{align*}
V\left(\hat{h}\right) & :=\int\ind\left\{ \hat{h}\left(x\right)\geq0\right\} v_0\left(x\right)f\left(x\right)dx\\
\equiv\ol V\left(\hat{\mu}\right) & :=\int\ind\left\{ \hat{\mu}\left(x,1\right)-\hat{\mu}\left(x,0\right)\geq0\right\} v_0\left(x\right)f\left(x\right)dx.
\end{align*}
Under Assumption \ref{assu:main} and \ref{assu:main-functional}(b)(c) the functional
derivative of $\ol V\left(\mu_{0}\right)\left[\nu\right]$ is given by 
\begin{align*}
D_\mu \ol V\left(\mu_{0}\right)\left[\nu\right] & :=\int_{\left\{x\in\R^d: h_{0}\left(x\right)=0\right\} }\frac{\left(\nu\left(x,1\right)-\nu\left(x,0\right)\right)}{\norm{\Dif_{x} h_{0}\left(x\right)}}v_0\left(x\right)f\left(x\right)d{\cal H}^{d-1}\left(x\right).
\end{align*}
As explained in Section \ref{sec:RootN}, $V\left(h_{0}\right)\equiv \ol V\left(\mu_{0}\right)$ is
generally not $\sqrt{n}$-estimable, in particular, the $D_\mu \ol V\left(\mu_{0}\right)\left[\nu\right]$ are not bounded (or continuous) linear functionals on $L^2(f)$, which means that they do not have a Riesz representer on the whole Hilbert space $L^2(f)$. Nevertheless, sieve Riesz representer is well-defined (see \citet*{chen2025semiparametric}), which will be the key ingredient for the sieve variance term for the properties of $\ol V\left(\hat{\mu}\right)-\ol V\left(\mu_{0}\right)$.
In particular we need to study
the asymptotic property of the plug-in estimator of the submanifold
integral of form \eqref{eq:Dh_Vh_form}, which has been studied in
\citet*{chen2025semiparametric} using linear series (Bspline) first
stage: in particular, Section 4.3 of \citet*{chen2025semiparametric}
analyzes the integral on upper contour set of the form $V\left(h_{0}\right)$
specifically. We use the result in \citet*{chen2025semiparametric}
without repeating it here, but focus on the adaptation required for
the standard error computation.
\begin{assumption}
\label{assu:Asym_Vh}Suppose that: 
\begin{itemize}
\item[(a)] $\norm{\Dif_x^{2}h_{0}\left(x\right)}\leq M<\infty$. 
\item[(b)] $\norm{\hat{\mu}-\mu_{0}}_{\infty}\norm{\Dif\left(\hat{\mu}-\mu_{0}\right)}_{\infty}=o_{p}\left(\sqrt{\frac{1}{n}K_{n}^{\frac{1}{d}}}\right).$
\end{itemize}
\end{assumption}

\begin{thm}
\label{thm:Vfunc_IntF} Suppose that Assumptions \ref{assu:main} and \ref{assu:main-functional}(b)(c) hold. Let $\hat{\mu}$ be a linear sieve estimator of $\mu_{0}$ and suppose that Assumptions 6, 8 and 11 in \citet*{chen2025semiparametric}
hold along with Assumption \ref{assu:Asym_Vh} above. Then:
\[
\frac{\sqrt{n}\left(\ol V\left(\hat{\mu}\right)-\ol V\left(\mu_{0}\right)\right)}{\s_{V,n}}\dto\cN\left(0,1\right),~~~\text{with }\s_{V,n}^{2}\asymp K_{n}^{\frac{1}{d}}
\]
\end{thm}

The standard error estimates can be computed based on the linear sieve first stage in a way similar to that described in Section 3.1, with the following adaptions. Again, we use the formula 
\begin{equation}\label{eq:sigma_v_est}
\hat{\s}_{V}^{2}:=\hat{D}_\mu \overline{V}\left(\hat{\mu}\right)\left[\overline{\psi}\right]'\hat{\O}\hat{D}_\mu \overline{V} \left(\hat{\mu}\right)\left[\overline{\psi}\right],
\end{equation}
where the pathwise derivative $D_\mu \ol{V}\left(\hat{\mu}\right)\left[\ol{\psi}\right]$ given by 
\[
D_\mu \ol{V}\left(\hat{\mu}\right)\left[\ol{\psi}\right] = 
\begin{pmatrix}
\int_{\{x \in {\cal X} : \hat{h}(x) = 0 \}} 
\frac{\psi^{(K_1)}(x)}{\lVert \nabla_x h_{0}(x) \rVert}
\,v_0(x) f(x) \, d\mathcal{H}^{d - 1}(x)\\
-\int_{\{x \in {\cal X} : \hat{h}(x) = 0 \}} 
\frac{\psi^{(K_0)}(x)}{\lVert \nabla_x h_{0}(x) \rVert}
\, v_0(x)f(x) \, d\mathcal{H}^{d - 1}(x)
\end{pmatrix},
\]
can be approximated via
\begin{equation}\label{eq:D_mu_V}
\hat{D}_\mu\overline{V}\left(\hat{\mu}\right)\left[\ol{\psi}\right] =
\begin{pmatrix}
\frac{1}{2\e}\int_{\{x \in {\cal X} :-\e < \hat{h}(x) < \e\}} 
\,\psi^{(K_1)}(x) v_0 (x) f(x) \, dx    \\
-\frac{1}{2\e}\int_{ \{x \in {\cal X}:-\e < \hat{h}(x) < \e\}} 
\,\psi^{(K_0)}(x) v_0 (x) f(x) \, dx 
\end{pmatrix},
\end{equation}
based on the mathematical result\footnote{See Theorem 3.13.(iii) of \citet*{evans2015measure}.}
that 
\[
\lim_{\e\downto0}\frac{1}{2\e}\int_{\left\{ x\in{\cal X}:~-\e<h\left(x\right)<\e\right\} }\omega\left(x\right)dx=\int_{\left\{ x\in{\cal X}:~h\left(x\right)=0\right\} }\frac{\omega\left(x\right)}{\norm{\Dif_{x}h\left(x\right)}}d{\cal H}^{d-1}\left(x\right).
\]
Again, we use Sobol points for numerical integration. See the simulation section for details, as well as robustness checks with respect the choice of $\e$ in the numerical differentiation step.

\subsection{\label{subsec:ValueFunc_mean}Value Functional Under 
Unknown Density $f=f_{0}$}

We now consider the case where $F=F_{0}$ and population expectation
$\E[\cdot]$ is estimated by the sample average $\frac{1}{n}\sum_{i=1}^{n}[\cdot]$
, and seek to characterize the asymptotic behavior of the natural
plug-in estimator of $V\left(h_{0}\right)\equiv\ol V\left(\mu_{0}\right)$
is given by
\begin{align*}
\hat{V}\left(\hat{h}\right) & :=\frac{1}{n}\sum_{i=1}^{n}\ind\left\{ \hat{h}\left(X_{i}\right)\geq0\right\} v_{0}\left(X_{i}\right)\\
\equiv\hat{\ol V}\left(\hat{\mu}\right) & :=\frac{1}{n}\sum_{i=1}^{n}\ind\left\{ \hat{\mu}\left(X_{i},1\right)-\hat{\mu}\left(X_{i},0\right)\geq0\right\} v_{0}\left(X_{i}\right).
\end{align*}

\noindent It turns out that the additional error in the approximation
of $V\left(\hat{h}\right)$ by $\hat{V}\left(\hat{h}\right)$ is asymptotically negligible relative to $V\left(\hat{h}\right)-V(h_0)$,
which converges at a slower-than-$\sqrt{n}$ rate.
\begin{thm}
\label{thm:Vfunc_Mean}The asymptotic distribution of $\hat{V}\left(\hat{h}\right)$ coincides with that $V(\hat{h})$ in Theorem \ref{thm:Vfunc_IntF}.
\end{thm}

The standard error can be computed using formula \eqref{eq:sigma_v_est}, with $\hat{D}_\mu\overline{V}\left(\hat{\mu}\right)\left[\ol{\psi}\right]$ given by the following adapted estimator,
\begin{equation}\label{eq:D_mu_V_fhat}
\hat{D}_\mu\overline{V}\left(\hat{\mu}\right)\left[\ol{\psi}\right] =
\begin{pmatrix}
\frac{1}{2\e}\int_{\{x \in {\cal X} :-\e < \hat{h}(x) < \e\}} 
\,\psi^{(K_1)}(x) v_0 (x) \hat{f}(x) \, dx    \\
-\frac{1}{2\e}\int_{ \{x \in {\cal X}:-\e < \hat{h}(x) < \e\}} 
\,\psi^{(K_0)}(x) v_0 (x) \hat{f}(x) \, dx 
\end{pmatrix},
\end{equation}
where $\hat{f}(x)$ is a nonparametric density estimator of $f_0(x)$. We then again use Sobol points to numerically evaluate the integral in \eqref{eq:D_mu_V_fhat}.

\begin{Remark}\label{rem:int_Sobol_fhat}
Even though integrals of the form $\int w(x) f_0(x) dx$ can be estimated using the sample average $\frac{1}{N}\sum_iw(X_i)$ without the need of a nonparametric density estimator, we choose instead to estimate it using $\int w(x) \hat{f}(x) dx$  using the nonparametric density estimator $\hat{f}(x)$ together with numerical integration based on Sobol points in \eqref{eq:D_mu_V_fhat}, because we need to compute the integral in \eqref{eq:D_mu_V_fhat} on a small constrained domain $\{-\e < \hat{h}(x)<\e\}$ for numerical differentiation. As a result, given the empirical data $(X_i)_{i=1}^N$, for small choices of $\e$, there might be very few, or even no, data points that fall into the small band, making the sample average $\frac{1}{N}\sum_iw(X_i)\ind\{-\e < \hat{h}(X_i)<\e\}$ very discrete (and even identically zero) for small values of $\e$. The use of a nonparametric density estimator $\hat{f}(x)$, along with numerical integration, effectively smooths out such discreteness and results in much better numerical approximation of the derivative.  
\end{Remark}

\section{Simulations}\label{sec:Sim}

\subsection{Results for Theorem \ref{thm:Welfare_intF}} \label{Thm1_Sim}

We first report the finite-sample performance of the semiparametric estimator, its associated standard error estimator, and the resulting confidence interval, based on the theoretical results in Theorem~\ref{thm:Welfare_intF}, which concern the welfare functional under a known target distribution. The model specifications used in the simulations are summarized in Table~\ref{tab:Thm1_model_specs}. For each specification, random samples of size $n$ are drawn with covariates $X_i \sim F_0$, and treatment status is assigned according to the propensity score function $p_0(X_i)$. Outcomes are then generated as $Y_i = \mu_0(X_i, D_i) + \epsilon_i$, with $\epsilon_i \sim N(0,1)$.

\begin{table}[!htbp]
\centering
\caption{Theorem 1: Model Specifications}
\label{tab:Thm1_model_specs}
\scriptsize
\begin{tabular}{@{}lccccc@{}}
\toprule
Model & \textbf{$F_0$} & \textbf{$F$} & \textbf{$\mu_0(x,d)$} & \textbf{$p_0(x)$} & \textbf{$\sigma^2$} \\
\midrule
M1 & $\text{U}[-0.2,1.2]$ & $\text{U}[0,1]$ & 
$\begin{aligned}5\sin(2\pi x)\cos(2\pi x)\\+ d(-0.4+2x^2)\end{aligned}$ & 
$\frac{1}{1+e^{-(1-2x)}}$ & 1 \\
\addlinespace
M2 & $\text{U}[-0.2,1.2]$ & $\text{U}[0,1]$ & 
$0.5|x|+d(0.5-x^2)$ & 
$\frac{1}{1+e^{-(-0.5+x)}}$ & 1 \\
\addlinespace
M3 & $\text{U}[-0.2,1.2]$ & $\text{U}[0,1]$ & 
$x^2+d(1-x)$ & 
$\frac{1}{1+e^{-(0.5-x)}}$ & 1 \\
\addlinespace
M4 & $\text{U}[-0.2,1.2]^2$ & $\text{U}[0,1]^2$ & 
$\begin{aligned}(1-x_1^2-x_2^2)(4+\sin x_1x_2+\cos x_2)\\+d(0.5x_1-0.4x_2)\end{aligned}$ & 
$\frac{1}{1+e^{-(x_1-x_2)}}$ & 1 \\
\addlinespace
M5 & $\text{U}[-0.2,1.2]^2$ & $\text{U}[0,1]^2$ & 
$\begin{aligned}(1-x_1x_2)(3+\sin(\pi x_1)\cos(\pi x_2))\\+d(0.3x_1-0.3x_2)\end{aligned}$ & 
$\frac{1}{1+e^{-(x_1-x_2)}}$ & 1 \\
\addlinespace
M6 & $\text{U}[-0.2,1.2]^2$ & $\text{U}[0,1]^2$ & 
$\log(1+x_1+x_2)+d(x_1-0.7x_2)$ & 
$\frac{1}{1+e^{-(1.5x_1-0.5x_2)}}$ & 1 \\
\addlinespace
M7 & $\text{U}[-0.2,1.2]^2$ & $\text{U}[0,1]^2$ & 
$\begin{aligned}(x_1^2+x_2^2)e^{-(x_1+x_2)}\\+d(0.5-x_2)\end{aligned}$ & 
$\frac{1}{1+e^{-(-0.5+x_1+2x_2)}}$ & 1 \\
\bottomrule
\end{tabular}
\end{table}

The semiparametric estimator of the welfare functional is constructed in two steps. In the first step, $\mu_0(x,1)$ and $\mu_0(x,0)$ are estimated separately for the treated group ($D_i = 1$) and the control group ($D_i = 0$) using B-spline regressions. In the second step, the welfare functional is approximated by numerically integrating over $M = 5{,}000$ Sobol points $\{X_j^{Sobol}\}_{j=1}^M$ drawn from the target distribution $F$:
\[
\hat{\overline{W}}(\hat{\mu}) \;=\; \frac{1}{M} \sum_{j=1}^M \big[\hat{\mu}(X_j^{Sobol},1) - \hat{\mu}(X_j^{Sobol},0)\big]_+,
\]
where $\hat{\mu}(x,1)$ and $\hat{\mu}(x,0)$ denote the first-stage nonparametric estimators. 

Let $\hat{h}(x) = \hat{\mu}(x,1) - \hat{\mu}(x,0)$ denote the estimated CATE function, and let $\hat{p}(x)$ be a B-spline sieve estimator of the propensity score. The $95\%$ confidence interval for the welfare functional takes the usual form,
\[
\left[\hat{\overline{W}}(\hat{\mu}) - 1.96 \frac{\hat{\sigma}_W}{\sqrt{n}}, \;\;  \hat{\overline{W}}(\hat{\mu}) + 1.96 \frac{\hat{\sigma}_W}{\sqrt{n}}\right],
\]
with
\[
\hat{\s}_{W}^{2} =\frac{1}{N}\sum_{i}\frac{\ind\left\{ \hat{h}\left(X_{i}\right)\geq0\right\} \l^{2}\left(X_{i}\right)(Y_i -\hat{h}(X_i))^2}{\hat{p}\left(X_{i}\right)\left(1-\hat{p}\left(X_{i}\right)\right)}.
\]

To evaluate performance, we simulate each model 2,000 times at sample sizes $n = 1500, 3000,$ and $6000$. Table~\ref{tab:Thm1_Sim_Result} reports the true welfare functional ($W_{\text{true}}$), the average bias of the estimator $\hat{\overline{W}}(\hat{\mu})$ (Bias), its sampling standard deviation (SD), the average estimated standard error (SE), the standard deviation of SE across iterations (SD(SE)), and the empirical coverage probability of the associated 95\% confidence interval (Coverage). The results\footnote{Additional simulations based on GAM are presented in the Appendix \ref{tab:Thm1_Sim_GAM}.} indicate that the nominal coverage rate is attained in nearly all cases, even with relatively small samples, and improves further as sample size increases. The bias of the estimator also becomes negligible relative to its sampling variability, and the proposed SE estimator closely tracks the sampling standard deviation with high precision.

\begin{table}[!htbp]
\centering
\caption{Theorem 1: Simulation Results}
\label{tab:Thm1_Sim_Result}
\begin{tabular}{lccccccc}
\toprule
Model & $n$ & $W_{\text{true}}$ & Bias & SD & SE & SD(SE) & Coverage \\
\midrule
M1 & 1500 & 0.3857 & 0.0161 & 0.0470 & 0.0480 & 0.0033 & 0.9415 \\
   & 3000 & 0.3857 & 0.0080 & 0.0338 & 0.0337 & 0.0018 & 0.9480 \\
   & 6000 & 0.3857 & 0.0037 & 0.0229 & 0.0237 & 0.0010 & 0.9535 \\
\midrule
M2 & 1500 & 0.2358 & 0.0030 & 0.0484 & 0.0518 & 0.0031 & 0.9615 \\
   & 3000 & 0.2358 & 0.0028 & 0.0352 & 0.0366 & 0.0015 & 0.9555 \\
   & 6000 & 0.2358 & 0.0009 & 0.0243 & 0.0259 & 0.0007 & 0.9630 \\
\midrule
M3 & 1500 & 0.5001 & 0.0046 & 0.0579 & 0.0605 & 0.0025 & 0.9620 \\
   & 3000 & 0.5001 & 0.0025 & 0.0407 & 0.0430 & 0.0013 & 0.9595 \\
   & 6000 & 0.5001 & 0.0000 & 0.0293 & 0.0305 & 0.0006 & 0.9625 \\
\midrule
M4 & 1500 & 0.1033 & 0.0185 & 0.0478 & 0.0561 & 0.0074 & 0.9735 \\
   & 3000 & 0.1033 & 0.0088 & 0.0348 & 0.0400 & 0.0040 & 0.9745 \\
   & 6000 & 0.1033 & 0.0043 & 0.0248 & 0.0284 & 0.0021 & 0.9695 \\
\midrule
M5 & 1500 & 0.0499 & 0.0282 & 0.0418 & 0.0521 & 0.0093 & 0.9700 \\
   & 3000 & 0.0499 & 0.0158 & 0.0303 & 0.0369 & 0.0056 & 0.9720 \\
   & 6000 & 0.0499 & 0.0094 & 0.0227 & 0.0262 & 0.0032 & 0.9660 \\
\midrule
M6 & 1500 & 0.2315 & 0.0268 & 0.0559 & 0.0622 & 0.0048 & 0.9550 \\
   & 3000 & 0.2315 & 0.0123 & 0.0394 & 0.0444 & 0.0025 & 0.9655 \\
   & 6000 & 0.2315 & 0.0050 & 0.0288 & 0.0315 & 0.0014 & 0.9695 \\
\midrule
M7 & 1500 & 0.1250 & 0.0462 & 0.0505 & 0.0595 & 0.0073 & 0.9370 \\
   & 3000 & 0.1250 & 0.0218 & 0.0346 & 0.0406 & 0.0043 & 0.9610 \\
   & 6000 & 0.1250 & 0.0101 & 0.0245 & 0.0278 & 0.0022 & 0.9680 \\
\bottomrule
\end{tabular}

\begin{flushleft}
\footnotesize \textit{Notes:} (1) $W_{\text{true}}$ denotes the true welfare functional, computed numerically using $5{,}000$ Sobol points drawn from $F$.  
(2) Bias is the average deviation from $W_{\text{true}}$.  
(3) SD is the sampling standard deviation.  
(4) SE is the average estimated asymptotic standard error of $\hat{\overline{W}}(\hat{\mu})$.  
(5) SD(SE) is the standard deviation of SE across iterations.  
(6) Coverage is the empirical coverage probability of the 95\% CI.  
\end{flushleft}
\end{table}


\begin{Remark}
    To improve computational efficiency, we predetermine the sieve dimensions for estimating $\mu_0(x,0)$ and $\mu_0(x,1)$. For each model specification, we first generate a dataset ${(Y_i, X_i, D_i)}_{i=1}^{n=6000}$ and apply an adaptation the sieve dimension selection procedure of \citet*{chen2025adaptive} separately to the treated and control groups.\footnote{The adaptation selects a larger sieve dimension than that selected by the CCK procedure (and the npiv R package) to achieve undersmoothing.} As demonstrated in their paper, this approach ensures that the resulting estimators of $\mu_0(x,1)$ and $\mu_0(x,0)$ converge at the fastest possible (i.e., minimax) rates in the sup-norm. The selected sieve dimensions are then used in the simulation designs with $n = 1500, 3000, 6000$. In principle, one could implement data-driven dimension selection within each simulation iteration, but doing so would substantially increase computational cost.
\end{Remark}

\begin{Remark}
    The construction of the $95\%$ confidence interval  requires estimating the propensity score function $p_0(x)$. To this end, we again use the B-spline sieve estimator, regressing treatment status on the covariates. The sieve dimension is similarly predetermined using an adaptation of \citet*{chen2025adaptive} predetermined based on the full dataset. A potential concern is that the fitted propensity score $\hat{p}(x)$ may take values outside the unit interval for some observations, which is likely to indicate a violation of the strict overlap assumption. Accordingly, when computing the asymptotic standard deviation $\hat{\sigma}_W$, we trim observations with estimated propensity scores lying outside $[0,1]$.
\end{Remark}

\subsubsection{Sieve Variance Estimation}

One potential drawback of the plug-in approach to estimating the asymptotic variance of the welfare functional is that it requires nonparametric estimation of the propensity score function $p_0(x)$, which can introduce additional sampling noise and numerical instability. As an alternative, one may employ the sieve variance estimator as in \cite*{chen2014sieveIrregular} or \citet*{chen2015sieve}, whose formula depends only on the pathwise derivative of the welfare functional and a linear regression of the outcome variable on the sieve basis. When the target distribution $F$ is known, the pathwise derivative of the welfare functional can be numerically approximated using Sobol points drawn from the target population combined with importance sampling. When the target distribution is unknown, the pathwise derivative can instead be evaluated directly at the observed data and approximated by a sample average. Simulation results and the empirical application based on the sieve variance estimator are shown in the Appendix \ref{SieveVar}.

\subsection{Results for Theorem \ref{thm:Welfare_mean}}
We now investigate the finite-sample performance of our estimation and inference procedure for the welfare functional in the case where the target distribution $F$ coincides with the population distribution $F_0$, though both remain unknown. The relevant model specifications are presented in Table~\ref{tab:Thm2_model_specs}. In contrast to the designs considered in Section~\ref{Thm1_Sim}, these specifications explicitly impose that $F$ and $F_0$ share identical supports.

\begin{table}[!htbp]
\centering
\caption{Theorem 2: Model Specifications}
\label{tab:Thm2_model_specs}
\scriptsize
\begin{tabular}{@{}lccccc@{}}
\toprule
Model & \textbf{$F_0$} & \textbf{$F$} & \textbf{$\mu_0(x,d)$} & \textbf{$p_0(x)$} & \textbf{$\sigma^2$} \\
\midrule
M8 & $\text{U}[0,1]$ & $\text{U}[0,1]$ &
$\begin{aligned}5\sin(2\pi x)\cos(2\pi x) \\+ d(-0.4+2x^2)\end{aligned}$ &
$\frac{1}{1+e^{-(1-2x)}}$ & 1 \\
\addlinespace
M9 & $\text{U}[0,1]$ & $\text{U}[0,1]$ &
$0.5|x|+d(0.5-x^2)$ &
$\frac{1}{1+e^{-(-0.5+x)}}$ & 1 \\
\addlinespace
M10 & $\text{U}[0,1]$ & $\text{U}[0,1]$ &
$x^2+d(1-x)$ &
$\frac{1}{1+e^{-(0.5-x)}}$ & 1 \\
\addlinespace
M11 & $\text{U}[0,1]^2$ & $\text{U}[0,1]^2$ &
$\begin{aligned}(1-x_1^2-x_2^2)(4+\sin x_1x_2+\cos x_2) \\+ d(0.5x_1-0.4x_2)\end{aligned}$ &
$\frac{1}{1+e^{-(x_1-x_2)}}$ & 1 \\
\addlinespace
M12 & $\text{U}[0,1]^2$ & $\text{U}[0,1]^2$ &
$\begin{aligned}(1-x_1x_2)(3+\sin(\pi x_1)\cos(\pi x_2)) \\+ d(0.3x_1-0.3x_2)\end{aligned}$ &
$\frac{1}{1+e^{-(x_1-x_2)}}$ & 1 \\
\addlinespace
M13 & $\text{U}[0,1]^2$ & $\text{U}[0,1]^2$ &
$\log(1+x_1+x_2)+d(x_1-0.7x_2)$ &
$\frac{1}{1+e^{-(1.5x_1-0.5x_2)}}$ & 1 \\
\addlinespace
M14 & $\text{U}[0,1]^2$ & $\text{U}[0,1]^2$ &
$\begin{aligned}(x_1^2+x_2^2)e^{-(x_1+x_2)} \\+ d(0.5-x_2)\end{aligned}$ &
$\frac{1}{1+e^{-(-0.5+x_1+2x_2)}}$ & 1 \\
\bottomrule
\end{tabular}
\end{table}

The plug-in estimator $\hat{\overline{W}}(\hat{\mu})$ proposed here differs from that in Section~\ref{Thm1_Sim} only in the second step: instead of using Sobol points to approximate the integral, we take the sample average of $\big[\hat{\mu}(X_i,1) - \hat{\mu}(X_i,0)\big]_+$ over the observed data:
\[
\hat{\overline{W}}(\hat{\mu}) \;=\; \frac{1}{n}\sum_{i=1}^n \big[\hat{\mu}(X_i,1) - \hat{\mu}(X_i,0)\big]_+.
\]

Relative to the known-$F$ case, the asymptotic variance of $\hat{\overline{W}}(\hat{h})$ includes an additional component, $\mathrm{Var}([h_0(X_i)]_+)$. We estimate the asymptotic standard deviation $\hat{\overline{\sigma}}_W$ by substituting the B-spline sieve estimates for the nuisance functions and replacing the population mean with its sample analog. As before, we restrict attention to observations with estimated propensity scores in $[0,1]$. A 95\% confidence interval is then given by
\[
\left[\hat{\overline{W}}(\hat{\mu}) - 1.96 \,\frac{\hat{\overline{\sigma}}_W}{\sqrt{n}}, \;\;  \hat{\overline{W}}(\hat{\mu}) + 1.96 \,\frac{\hat{\overline{\sigma}}_W}{\sqrt{n}}\right].
\]

The simulation results\footnote{Additional simulations based on GAM are presented in the Appendix \ref{tab:Thm2_Sim_GAM}.} based on 2,000 iterations with sample sizes $n = 1500, 3000,$ and $6000$ are reported in Table~\ref{tab:Thm2_Sim}. Overall, the coverage of the proposed confidence intervals converges to the nominal 95\% level as sample size increases, while the decreasing bias and standard deviation indicate a reduction in mean squared error. 

\begin{Remark}
Extrapolation bias may arise when B-spline fits are evaluated outside the support of their training samples-for instance, when $\hat{\mu}(x,1)$ is evaluated using control group data or $\hat{\mu}(x,0)$ using treated group data. To mitigate this issue, we trim observations that fall outside the common support of treated and control groups, estimate the nuisance functions on the trimmed sample, and then compute the welfare functional. Since only a small fraction of observations are removed, this adjustment has a negligible effect on the results.
\end{Remark}

\begin{table}[!htbp]
\centering
\caption{Theorem 2: Simulation Results}
\label{tab:Thm2_Sim}
\begin{tabular}{lccccccc}
\toprule
Model & $n$ & $W_{\text{true}}$ & Bias & SD & SE & SD(SE) & Coverage \\
\midrule
M8 & 1500 & 0.3857 & 0.0068 & 0.0414 & 0.0423 & 0.0025 & 0.9515 \\
   & 3000 & 0.3857 & 0.0029 & 0.0296 & 0.0297 & 0.0013 & 0.9550 \\
   & 6000 & 0.3857 & 0.0006 & 0.0206 & 0.0210 & 0.0007 & 0.9575 \\
\midrule
M9 & 1500 & 0.2358 & 0.0042 & 0.0425 & 0.0440 & 0.0025 & 0.9605 \\
   & 3000 & 0.2358 & 0.0029 & 0.0304 & 0.0311 & 0.0012 & 0.9590 \\
   & 6000 & 0.2358 & 0.0012 & 0.0209 & 0.0221 & 0.0006 & 0.9555 \\
\midrule
M10 & 1500 & 0.5001 & 0.0067 & 0.0511 & 0.0515 & 0.0018 & 0.9495 \\
   & 3000 & 0.5001 & 0.0037 & 0.0357 & 0.0366 & 0.0009 & 0.9600 \\
   & 6000 & 0.5001 & 0.0013 & 0.0253 & 0.0260 & 0.0005 & 0.9600 \\
\midrule
M11 & 1500 & 0.1033 & 0.0307 & 0.0365 & 0.0402 & 0.0038 & 0.9245 \\
   & 3000 & 0.1033 & 0.0156 & 0.0264 & 0.0286 & 0.0021 & 0.9395 \\
   & 6000 & 0.1033 & 0.0084 & 0.0192 & 0.0203 & 0.0012 & 0.9470 \\
\midrule
M12 & 1500 & 0.0499 & 0.0418 & 0.0316 & 0.0373 & 0.0045 & 0.8790 \\
   & 3000 & 0.0499 & 0.0232 & 0.0229 & 0.0263 & 0.0029 & 0.9135 \\
   & 6000 & 0.0499 & 0.0126 & 0.0168 & 0.0186 & 0.0017 & 0.9355 \\
\midrule
M13 & 1500 & 0.2315 & 0.0177 & 0.0432 & 0.0459 & 0.0032 & 0.9565 \\
   & 3000 & 0.2315 & 0.0088 & 0.0309 & 0.0323 & 0.0014 & 0.9510 \\
   & 6000 & 0.2315 & 0.0041 & 0.0221 & 0.0229 & 0.0007 & 0.9525 \\
\midrule
M14 & 1500 & 0.1250 & 0.0251 & 0.0382 & 0.0421 & 0.0059 & 0.9460 \\
   & 3000 & 0.1250 & 0.0117 & 0.0258 & 0.0287 & 0.0030 & 0.9585 \\
   & 6000 & 0.1250 & 0.0054 & 0.0182 & 0.0198 & 0.0013 & 0.9595 \\
\bottomrule
\end{tabular}

\begin{flushleft}
\footnotesize \textit{Notes:} (1) $W_{\text{true}}$ is the true welfare functional, computed numerically using $5{,}000$ Sobol points drawn from $F$.  
(2) Bias is the average deviation from $W_{\text{true}}$.  
(3) SD is the sampling standard deviation.  
(4) SE is the average estimated asymptotic standard error of $\hat{\overline{W}}(\hat{\mu})$.  
(5) SD(SE) is the standard deviation of SE across iterations.  
(6) Coverage is the empirical coverage probability of the 95\% CI.  
\end{flushleft}
\end{table}

\subsection{Results for Theorem \ref{thm:Vfunc_IntF}}

To assess the finite-sample properties of our estimation inference procedure for the value functional $\overline{V}(\hat{\mu})$ under a known target distribution $F$, we analyze the model described in Table \ref{tab:Thm3_model_specs}: 
\begin{table}[H]
\centering
\caption{Theorem 3: Model Specification}
\label{tab:Thm3_model_specs}
\scriptsize
\begin{tabular}{@{}lcccccc@{}}
\toprule
Model & \textbf{$F_0$} & \textbf{$F$} & \textbf{$\mu_0(x_1,x_2,d)$} & \textbf{$\nu_0(x_1,x_2)$} & \textbf{$p_0(x_1,x_2)$} & \textbf{$\sigma^2$} \\
\midrule
M15 
& $\text{U}([-2,2]^2)$ 
& $\text{U}([-1.5,1.5]^2)$ 
& $\begin{aligned} d \cdot (1 - x_1^2 - x_2^2) \cdot (4 + \sin(x_1)x_2 + \cos(x_2)) \end{aligned}$ 
& $1$ 
& $\tfrac{1}{1 + e^{-(x_1 - x_2)}}$ 
& $1$ \\
\bottomrule
\end{tabular}
\end{table}

Our parameter of interest is a scaled value functional under known $F$:
\[
V(h_0) \;=\; 3^2 \int \mathbbm{1}\!\left\{ \big(1 - x_{1}^2 - x_{2}^2\big)\,\big(4 + \sin(x_{1})x_{2} + \cos(x_{2})\big) \;\geq\; 0 \right\} \, dF(x_{1},x_{2}),
\]
where the scaling factor $3^2$ is chosen so that the integral evaluates to $\pi$. Equivalently, this setup can be interpreted as a Monte Carlo experiment for estimating the area of the unit circle by uniformly throwing darts over a $3 \times 3$ square.

The plug-in estimator $\hat{\overline{V}}(\hat{\mu})$ for $\overline{V}(\mu_0)$ is constructed in two steps. In the first step, we estimate $\mu_0(x,1) = \mu_0(x_1,x_2,1)$ and $\mu_0(x,0) = \mu_0(x_1,x_2,0)$ separately using B-spline sieve estimators, fitted on the treated and control groups, respectively, with sieve dimensions predetermined as described in subsection~\ref{Thm1_Sim}. In the second step, $\mathbbm{1}\{(\hat{\mu}_0(x, 1) - \hat{\mu}_0(x, 0)) \geq 0\}$ is numerically integrated using $5000$ Sobol points drawn from $F$. The resulting semiparametric two-step estimator is
\[
\hat{\overline{V}}(\hat{\mu}) \;=\; \frac{1}{M} \sum_{j=1}^M \mathbbm{1}\{\hat{\mu}(X_j^{Sobol},1) - \hat{\mu}(X_j^{Sobol},0)\big\}.
\]

The 95\% confidence interval is then constructed as 
\[
\left[\hat{\overline{V}}(\hat{\mu}) - 1.96 \,\frac{\hat{\sigma}_V}{\sqrt{n}}, \;\;  \hat{\overline{V}}(\hat{\mu}) + 1.96 \,\frac{\hat{\sigma}_V}{\sqrt{n}}\right],
\]
where the asymptotic variance estimate $\hat{\s}_{V}^{2}$ is given by 
\[
\hat{\s}_{V}^{2}= \hat{D}_\mu \overline{V}\left(\hat{\mu}\right)\left[\ol{\psi}\right]'\hat{\O} \hat{D}_\mu \overline{V} \left(\hat{\mu}\right)\left[\ol{\psi}\right],
\]
with $\hat{\Omega}$ being the estimated asymptotic covariance matrix for the OLS estimators in the linear regression model \ref{Sieve_OLS}, and 
\[
\hat{D}_\mu\overline{V}\left(\hat{\mu}\right)\left[\ol{\psi}\right] =
3^2 \begin{pmatrix}
\frac{1}{2\e}\int_{\{x \in [-1.5, 1.5]^2 :-\e < \hat{h}(x) < \e\}} 
\,\psi^{(K_1)}(x) \frac{1}{3^2} \, dx    \\
-\frac{1}{2\e}\int_{ \{x \in [-1.5, 1.5]^2:-\e < \hat{h}(x) < \e\}} 
\,\psi^{(K_0)}(x) \frac{1}{3^2} \, dx 
\end{pmatrix}.
\]

We set the tuning parameter $\epsilon = 0.005$ to mitigate bias in $\hat{\overline{V}}(\hat{\mu})$ and approximate $\hat{D}_\mu \overline{V}(\hat{\mu})[\nu]$ using the sample average over $M$ Sobol draws from $F$. Because draws from $F$ are unlikely to fall within the set $\{x \in [-1.5, 1.5]^2 : -\epsilon < \hat{h}(x) < \epsilon\}$ when $\epsilon$ is small, we use $M = 1{,}000{,}000$ Sobol points to ensure accuracy. Simulation results\footnote{Additional simulations under alternative model specifications, as well as robustness checks for different values of $\epsilon$, are presented in Appendix~\ref{tab:Thm3_robust}.} based on 2,000 iterations with sample sizes $n = 1500, 3000,$ and $6000$ are reported in Table~\ref{tab:Thm3_Sim_Results}. The results show that the coverage rate reaches the nominal level with relatively modest sample sizes, even though the value functional is not $\sqrt{n}$-estimable. Also, as sample size increases, both bias and standard error decrease, and the plug-in estimator for the standard error provides a close estimate of the sampling standard deviation.

\begin{table}[!htbp]
\centering
\caption{Theorem 3: Simulation Results}
\label{tab:Thm3_Sim_Results}
\begin{tabular}{lccccccc}
\toprule
Model & $n$ & $V_{\text{true}}$ & Bias & SD & SE & SD(SE) & Coverage \\
\midrule
M15 & 1500 & 3.1416 & 0.0076 & 0.0710 & 0.0711 & 0.0092 & 0.9420 \\
    & 3000 & 3.1416 & 0.0080 & 0.0486 & 0.0499 & 0.0029 & 0.9475 \\
    & 6000 & 3.1416 & 0.0062 & 0.0337 & 0.0353 & 0.0016 & 0.9490 \\
\bottomrule
\end{tabular}

\begin{flushleft}
\footnotesize \textit{Notes:} (1) $V_{\text{true}}$ is the true value functional, equal to $\pi$.  
(2) Bias is the average deviation from $V_{\text{true}}$.  
(3) SD is the sampling standard deviation of the estimator.  
(4) SE is the average estimated asymptotic standard error across iterations.  
(5) SD(SE) is the standard deviation of SE across iterations.  
(6) Coverage is the empirical coverage probability of the $95\%$ confidence interval.  
\end{flushleft}
\end{table}

\section{Empirical Application}\label{sec:Emp}

We revisit the empirical application analyzed in \citet[KT18]{kitagawa2018should} using Job Training Parternship Act (JTPA) dataset. The JTPA study randomized whether applicants were eligible to receive a mix of training, job-search assistance, and other services provided under the program for a period of 18 months. A detailed description of the study and an assessment of average program effects for five major subgroups of the target population are provided in \citet*{BloomEtAl1997}.

We evaluate welfare using two outcome measures, following the approach in KT18. The first is total earnings over the 30 months following treatment assignment. The second adjusts this measure by subtracting \$774 for individuals assigned to treatment, thereby incorporating program costs. These outcomes are considered from an intention-to-treat perspective, meaning we focus on eligibility assignment rather than treatment effects among compliers. The available covariates include applicants' pre-program earnings, years of education, and treatment status. Our objective is to estimate and conduct inference on the first-best welfare functional and the optimal fraction of the population that should receive treatment.

As the first step in the estimation and inference procedure, we trim observations outside the common support of the treated and control groups to enforce the overlap assumption and to avoid extrapolation when applying sieve estimators. For either outcome measure, we then estimate $\mu_0(x,1)$ nonparametrically using B-spline sieves fitted on the treated sample, and $\mu_0(x,0)$ analogously on the control sample. The sieve dimensions are selected in a data-driven manner following \citet*{chen2025adaptive}. Combining these estimates on the common support yields the CATE estimate $\hat{h}(x) = \hat{\mu}(x,1) - \hat{\mu}(x,0)$. Taking sample averages of $[\hat{h}(x)]_+$ and $\mathbbm{1}\{\hat{h}(x) > 0\}$ over the trimmed dataset produces the estimators $\hat{\overline{W}}(\hat{\mu})$ and $\hat{\overline{V}}(\hat{\mu})$, respectively. Confidence intervals for $\overline{W}(\mu_0)$ and $\overline{V}(\mu_0)$ are then constructed according to their asymptotic theories: $\left(\hat{\overline{W}}(\hat{\mu}) \pm 1.96 \,\hat{\ol{\sigma}}_W/\sqrt{N}\right),$ and $\left(\hat{\overline{V}}(\hat{\mu}) \pm 1.96 \,\hat{\sigma}_V/\sqrt{N}\right)$, where $N$ denotes the sample size of the JTPA dataset.

While computing sieve estimate of $\hat{\overline{\sigma}}_W$ is straightforward, computing sieve estimate of $\hat{\sigma}_V$ involves several additional steps. Specifically, it requires a density estimate $\hat{f}(x)$ for the covariates as discussed in Remark \ref{rem:int_Sobol_fhat}. To this end, we use a Gaussian kernel density estimator, selecting the bandwidth matrix via the smoothed cross-validation method (\texttt{Hscv()} in the \texttt{ks} package) and applying a scaling factor of $s=3$ to ensure adequate smoothness in the presence of discrete values for years of education. Furthermore, we need to specify a small hyperparameter $\epsilon$ to provide a close approximate for the pathwise derivative of the value functional in $\hat{\sigma}_V$. We set $\epsilon$ equal to a fraction $\iota = 0.01$ of the standard deviation of $\hat{h}(x)$ over the trimmed dataset. Robustness checks on the tuning parameters $s$ and $\iota$ are reported in Appendix~\ref{Appendix_Emp_Tuning}, and alternative density estimation approaches are examined in Appendix~\ref{Appendix_Emp_Density}.

Table~\ref{tab:emp_compare1} presents our estimation and inference results alongside the corresponding findings from KT18. Specifically, we report our estimated welfare gain and the share of individuals to be treated, along with their confidence intervals, based on the trimmed dataset. For comparability, we also present results obtained using the untrimmed dataset-on which KT18 conducted their analysis-with the same choice of tuning parameters. The nonparametric plug-in estimates from KT18, which target the same welfare and share parameters as in our analysis, serve as an empirical benchmark. In addition, the linear rule estimates with their associated confidence intervals from KT18 are reported to assess how conservative our nonparametric inference procedure is relative to their parametric counterparts.

Across both the trimmed and untrimmed datasets, our estimates of the welfare gain and the optimal treatment share are broadly consistent with the nonparametric plug-in rule estimates reported in KT18, despite methodological difference in the first stage: we employ sieve estimators for the nuisance functions, whereas they use a Nadaraya-Watson estimator with an Epanechnikov kernel. A key contribution of our analysis is the provision of confidence intervals for both parameters of interest. By contrast, KT18 report confidence intervals only for the welfare gain under parametric rules, but not for the optimal treatment share. Although our confidence interval for the welfare gain appears wide, its length is comparable to that under the linear rule in KT18, underscoring that our inference procedure remains sharp even while relying on fully nonparametric methods.


\begin{table}[htbp]
\centering
\caption{Estimated Welfare Gains and Share of Population to be Treated Under Nonparametric Plug-in Rule}
\label{tab:emp_compare1}
\begin{subtable}{\textwidth}
\centering
\caption{30-Month Post-Program Earnings, No Treatment Cost}
\begin{tabular}{@{}p{6cm}cc@{}}
\toprule
\textbf{Method} & \textbf{Share Treated} & \textbf{Est. Welfare Gain} \\
\midrule
Ours  (With Trimming)       & \makecell{0.89\\ \scriptsize{(0.73, 1.05)}} & \makecell{\$1,519  \\ \scriptsize{(\$691, \$2347)}}
\\
\midrule
Ours (Without Trimming)         & \makecell{0.92\\ \scriptsize{(0.76, 1.08)}} & \makecell{\$1,459  \\ \scriptsize{(\$840, \$2078)}}
\\
\midrule
KT18 (Nonparametric Plug-in)  & \makecell{0.91 \\ \scriptsize{NA}} & \makecell{\$1,693  \\ \scriptsize{NA}} \\
\midrule
KT18 (Linear)  & \makecell{0.96 \\ \scriptsize{NA}} & \makecell{\$1,180 \\ \scriptsize{(\$464, \$1,896)}} \\
\bottomrule
\end{tabular}
\end{subtable}

\vspace{1em}

\begin{subtable}{\textwidth}
\centering
\caption{30-Month Post-Program Earnings, \$774 Cost per Treatment}
\begin{tabular}{@{}p{6cm}cc@{}}
\toprule
\textbf{Method} & \textbf{Share Treated} & \textbf{Est. Welfare Gain} \\
\midrule
Ours (With Trimming)       & \makecell{0.80\\ \scriptsize{(0.53, 1.07)}} & \makecell{\$858  \\ \scriptsize{(\$152, \$1564)}}
\\
\midrule
Ours (Without Trimming)         & \makecell{0.85 \\ \scriptsize{(0.65, 1.05)}} & \makecell{\$768  \\ \scriptsize{(\$190, \$1347)}}
\\
\midrule
KT18 (Nonparametric Plug-in) & \makecell{0.78 \\ \scriptsize{NA}} & \makecell{\$996 \\ \scriptsize{NA}} \\
\midrule
KT18 (Linear) & \makecell{0.69 \\ \scriptsize{NA}} & \makecell{404 \\ \scriptsize{(\$-313,\$1,121)}} \\
\bottomrule
\end{tabular}
\end{subtable}

\vspace{1em}
\begin{flushleft}
\small \textit{Note:} We present point estimates and confidence intervals based on both the trimmed and untrimmed datasets, along with the linear rule estimates (with its confidence intervals) and the nonparametric plug-in rule estimates from \cite*{kitagawa2018should}. 
\end{flushleft}
\end{table}

\newpage

\bibliographystyle{ecca}
\bibliography{OptTreat}

\appendix

\section{Main Proofs}\label{app:Proof}

\begin{proof}[Proof of Theorem \ref{thm:Welfare_intF}]
For any $\nu$ s.t. $\int\nu^{2}\left(x\right)f\left(x\right)<\infty,$write
$h_{t}:=h_{0}+t\nu$ and consider the functional derivative
\begin{align*}
D_{h}\ol W\left(h_{0}\right)\left[\nu\right] & =\rest{\frac{d}{dt}\ol W\left(h_{t}\right)}_{t=0}\\
 & =\rest{\frac{d}{dt}\int\left[h_{t}\left(x\right)\right]_{+}f\left(x\right)dx}_{t=0}.
\end{align*}
We first show the interchangeability of the differentiation and integration
holds:
\[
\rest{\frac{d}{dt}\int\left[h_{t}\left(x\right)\right]_{+}f\left(x\right)dx}_{t=0}=\int\rest{\frac{\p}{\p t}\left[h_{t}\left(x\right)\right]_{+}}_{t=0}f\left(x\right)dx.
\]
Notice that $\left[h_{t}\left(x\right)\right]_{+}$ is Lipchitz in
$t$, we have
\[
\left|\left[h_{t}\left(x\right)\right]_{+}-\left[h_{s}\left(x\right)\right]_{+}\right|\leq\left|\nu\left(x\right)\right|\left|t-s\right|
\]
with 
\[
\int\left|\nu\left(x\right)\right|f\left(x\right)dx<\infty.
\]
Furthermore, $\left[h_{t}\left(x\right)\right]_{+}$ is almost everywhere
differentiable in $t$ with
\begin{equation}
\rest{\frac{\p}{\p t}\left[h_{t}\left(x\right)\right]_{+}}_{t=0}=\ind\left\{ h_{0}\left(x\right)\geq0\right\} \nu\left(x\right)\label{eq:dif_relu}
\end{equation}
for any $x$ s.t. $h_{0}\left(x\right)\neq0$. Since $\left\{ x:h\left(x\right)=0\right\} $
has Lebesgue measure $0$, we have
\[
P_{f}\left(h_{0}\left(X_{i}\right)=0\right)=0.
\]
Hence, \eqref{eq:dif_relu} holds a.s.-$f$ in $x$, and thus by the
dominated convergence theorem, we have
\begin{align}
\rest{\frac{d}{dt}\int\left[h_{t}\left(x\right)\right]_{+}f\left(x\right)dx}_{t=0} & =\int\rest{\frac{\p}{\p t}\left[h_{t}\left(x\right)\right]_{+}}_{t=0}f\left(x\right)dx.\nonumber \\
 & =\int\ind\left\{ h_{0}\left(x\right)\geq0\right\} \nu\left(x\right)f(x)dx\label{eq:dt_int_relu}
\end{align}

We now switch the notation from $h_{0}$ to $\mu_{0}$ for subsequent
analysis, and consider the functional derivative of $\ol W\left(\mu_{0}\right)$
w.r.t. $\mu$ in the direction of $\nu$. Note that $\mu\left(x,d\right)$
and $\nu\left(x,d\right)$ are functions of both $x$ and $d$. Applying
\eqref{eq:dt_int_relu}, we have
\begin{align}
D_{\mu}\ol W\left(\mu_{0}\right)\left[\nu\right]:=\  & \rest{\frac{d}{dt}\ol W\left(\mu_{0}+t\nu\right)}_{t=0}\nonumber \\
=\  & \int\ind\left\{ \mu_{0}\left(x,1\right)-\mu_{0}\left(x,0\right)\geq0\right\} \left(\nu\left(x,1\right)-\nu\left(x,0\right)\right)f\left(x\right)dx\nonumber \\
=\  & \int\ind\left\{ h_{0}\left(x\right)\geq0\right\} \left(\nu\left(x,1\right)-\nu\left(x,0\right)\right)\l\left(x\right)f_{0}\left(x\right)dx\nonumber \\
=\  & \E\left[\ind\left\{ h_{0}\left(X_{i}\right)\geq0\right\} \l\left(X_{i}\right)\left(\nu\left(X_{i},1\right)-\nu\left(X_{i},0\right)\right)\right]\nonumber \\
=\  & \E\left[\ind\left\{ h_{0}\left(X_{i}\right)\geq0\right\} \l\left(X_{i}\right)\left(\frac{D_{i}}{p_{0}\left(X_{i}\right)}\nu\left(X_{i},D_{i}\right)-\frac{1-D_{i}}{1-p_{0}\left(X_{i}\right)}\nu\left(X_{i},D_{i}\right)\right)\right]\nonumber \\
=\  & \E\left[\ind\left\{ h_{0}\left(X_{i}\right)\geq0\right\} \l\left(X_{i}\right)\left(\frac{D_{i}}{p_{0}\left(X_{i}\right)}-\frac{1-D_{i}}{1-p_{0}\left(X_{i}\right)}\right)\nu\left(X_{i},D_{i}\right)\right]\nonumber \\
=\  & \E\left[\nu^{*}\left(X_{i},D_{i}\right)\nu\left(X_{i},D_{i}\right)\right]\label{eq:DW_mu0}
\end{align}
where $\l:=f/f_{0}$, $\E[\cdot]$ is expectation taken with respect to the training data distribution, and
\begin{equation}
\nu^{*}\left(x,d\right):=\ind\left\{ h_{0}\left(x\right)\geq0\right\} \l\left(x\right)\left(\frac{d}{p_{0}\left(x\right)}-\frac{1-d}{1-p_{0}\left(x\right)}\right)\label{eq:DW_riesz}
\end{equation}
is the Riesz representer for the linear functional $D_{\mu}\ol W\left(\mu_{0}\right)\left[\cd\right]$. Since 
\begin{align*}
 & \E\left[\rest{\left(\frac{D_{i}}{p_{0}\left(X_{i}\right)}-\frac{1-D_{i}}{1-p_{0}\left(X_{i}\right)}\right)^{2}}X_{i}\right]\\
=\  & \E\left[\rest{\left(\frac{D_{i}-D_{i}p_{0}\left(X_{i}\right)-p_{0}\left(X_{i}\right)+D_{i}p_{0}\left(X_{i}\right)}{p_{0}\left(X_{i}\right)\left(1-p_{0}\left(X_{i}\right)\right)}\right)^{2}}X_{i}\right]\\
=\  & \frac{\E\left[\rest{\left(D_{i}-p_{0}\left(X_{i}\right)\right)^{2}}X_{i}\right]}{p_{0}^{2}\left(X_{i}\right)\left(1-p_{0}\left(X_{i}\right)\right)^{2}}=\frac{p_{0}\left(X_{i}\right)\left(1-p_{0}\left(X_{i}\right)\right)}{p_{0}^{2}\left(X_{i}\right)\left(1-p_{0}\left(X_{i}\right)\right)^{2}}\\
=\  & \frac{1}{p_{0}\left(X_{i}\right)\left(1-p_{0}\left(X_{i}\right)\right)}
\end{align*}
under Assumption \ref{assu:main}(b), the Riesz representer
$\nu^{*}$ has finite norm 
\begin{align*}
\norm{\nu^{*}}^{2} & =\E\left[\ind\left\{ h_{0}\left(X_{i}\right)\geq0\right\} \l^{2}\left(X_{i}\right)\left(\frac{D_{i}}{p_{0}\left(X_{i}\right)}-\frac{1-D_{i}}{1-p_{0}\left(X_{i}\right)}\right)^{2}\right]\\
 & =\E\left[\frac{\ind\left\{ h_{0}\left(X_{i}\right)\geq0\right\} \l^{2}\left(X_{i}\right)}{p_{0}\left(X_{i}\right)\left(1-p_{0}\left(X_{i}\right)\right)}\right]<\infty
\end{align*}
showing that $D_{\mu} \ol W\left(\mu_{0}\right)$ is a regular linear functional.

Next, we control the remainder term from the linearization. Specifically,
note that 
\[
D_{h}W\left(h\right)\left[\nu\right]=\int\ind\left\{ h\left(x\right)\geq0\right\} \nu\left(x\right)f(x) dx
\]
which is exactly of the form of the generic value functional. We then
have
\begin{align*}
D_{h}^{2}W\left(h_{0}\right)\left[\nu,u\right] & =\int_{\left\{x\in\R^d: h_0\left(x\right)=0\right\} }\frac{1}{\norm{\Dif_x h_{0}\left(x\right)}}\nu\left(x\right)u\left(x\right)f(x)d{\cal H}^{d-1}\left(x\right)\\
 & \leq\frac{1}{\ul{\e}}\norm{\nu}_{\infty}\norm u_{\infty}.
\end{align*}
Similar bound applies to $D_{\mu}^{2}\ol W\left(h_{0}\right)\left[\nu,u\right]$ as well. Thus we obtain:
\[
\left|\ol W\left(\hat{\mu}\right)-\ol W\left(\mu_{0}\right)-D_{\mu} \ol W\left(\mu_{0}\right)\left[\hat{\mu}-\mu_{0}\right]\right|\leq M\norm{\hat{\mu}-\mu_{0}}_{\infty}^{2}.
\]
Hence, we have
\begin{align*}
\sqrt{n}\left(\ol W\left(\hat{\mu}\right)-\ol W\left(\mu_{0}\right)\right) & =\sqrt{n}D_{\mu}\ol W\left(\mu_{0}\right)\left[\hat{\mu}-\mu_{0}\right]+\sqrt{n}O_{p}\left(\norm{\hat{\mu}-\mu_{0}}_{\infty}^{2}\right)\\
 & =\frac{1}{\sqrt{n}}\sum_{i=1}^{n}\nu^{*}\left(X_{i},D_{i}\right)\e_{i}+o_{p}\left(1\right)+\sqrt{n}o_{p}\left(n^{-1/2}\right)\\
 & \dto\cN\left(0,\s_{W}^{2}\right)
\end{align*}
where, writing $\s_{\e}^{2}\left(x\right):=\E\left[\rest{\e_{i}^{2}}X_{i}=x\right]$,
\begin{align*}
\s_{W}^{2} & :=\text{Var}\left(\nu^{*}\left(X_{i},D_{i}\right)\e_{i}\right)=\E\left[\frac{\ind\left\{ h_{0}\left(X_{i}\right)\geq0\right\} \l^{2}\left(X_{i}\right)\s_{\e}^{2}\left(X_{i}\right)}{p_{0}\left(X_{i}\right)\left(1-p_{0}\left(X_{i}\right)\right)}\right]\\
 & =\int\frac{\ind\left\{ h_{0}\left(x\right)\geq0\right\} \l^{2}\left(x\right)\s_{\e}^{2}\left(x\right)}{p_{0}\left(x\right)\left(1-p_{0}\left(X_{i}\right)\right)}f_{0}(x)dx\\
 & =\int\frac{\ind\left\{ h_{0}\left(x\right)\geq0\right\} \l\left(x\right)\s_{\e}^{2}\left(x\right)}{p_{0}\left(x\right)\left(1-p_{0}\left(X_{i}\right)\right)}f(x)dx
\end{align*}
\end{proof}
\begin{proof}[Proof of Theorem \ref{thm:Welfare_mean}]
We apply the derivations in Theorem \ref{thm:Welfare_intF} with
$f=f_{0}$ and consequently $\l\equiv1$. 

Consider the following standard empircal process decomposition 
\begin{align}
\sqrt{n}\left(\P_{n}g_{\mu}-Pg_{\mu_{0}}\right) & =\sqrt{n}P\left(g_{\hat{\mu}}-g_{\mu_{0}}\right)+\GG_{n}g_{\mu_{0}}+\GG_{n}\left(g_{\hat{\mu}}-g_{\mu_{0}}\right),\label{eq:Emp_Decomp}
\end{align}
where $g_{\mu}\left(x\right):=\left[\mu\left(x,1\right)-\mu\left(x,0\right)\right]_{+}$,
$\P_{n}g:=\frac{1}{n}\sum_{i=1}^{n}g$, $Pg=\int gdF_{0}$ and $\GG_{n}:=\sqrt{n}\left(\P_{n}-P\right)$.
Note that the term $\sqrt{n}P\left(g_{\hat{h}}-g_{h_{0}}\right)$
corresponds the analysis of population expectation (integral) with
respect to the true distribution $F=F_{0}$ in the last subsection.
There are two additional terms that appear due to the use of the sample
average: a ``stochastic equicontinuity'' term $\GG_{n}\left(g_{\hat{\mu}}-g_{\mu_{0}}\right)$
that will be shown to be asymptotically negligible under the permanance
of Donsker property, as well as a ``CLT''-term $\GG_{n}g_{\mu_{0}}$
that adds to the asymptotic variance of the estimator. 

Under Assumption \ref{assu:main}(c), $\hat{\mu}-\mu_{0}$
is assumed to belong to a Donsker class of functions, and the Lipchitz
transformation of $\mu_{0}$ through the ReLU function $\left[\cd\right]_{+}$
preserves the Donsker property. Hence, $g_{\hat{\mu}}-g_{\mu_{0}}$
also belongs to a Donsker class and thus $\GG_{n}\left(g_{\hat{\mu}}-g_{\mu_{0}}\right)=o_{p}\left(1\right)$.

Then, by \eqref{eq:Emp_Decomp} we have, 
\begin{align*}
 & \sqrt{n}\left(\hat{W}\left(\hat{h}\right)-W\left(h_{0}\right)\right)\equiv\sqrt{n}\left(\hat{\ol W}\left(\hat{\mu}\right)-\ol W\left(\mu_{0}\right)\right)\\
= & \frac{1}{\sqrt{n}}\sum_{i=1}^{n}\left(\left[h_{0}\left(X_{i}\right)\right]_{+}-W\left(h_{0}\right)+\nu^{*}\left(X_{i},D_{i}\right)\e_{i}\right)+o_{p}\left(1\right)\dto\cN\left(0,\ol{\s}_{W}^{2}\right)
\end{align*}
where 
\begin{align*}
\ol{\s}_{W}^{2} & :=\text{Var\ensuremath{\left(\left[h_{0}\left(X_{i}\right)\right]_{+}-W(h_0)+\nu^{*}\left(X_{i},D_{i}\right)\e_{i}\right)}}\\
  & =\text{Var\ensuremath{\left(\left[h_{0}\left(X_{i}\right)\right]_{+}\right)}}+\text{Var}\left(\nu^{*}\left(X_{i},D_{i}\right)\e_{i}\right)+2\text{Cov}\left(\left[h_{0}\left(X_{i}\right)\right]_{+},\nu^{*}\left(X_{i},D_{i}\right)\e_{i}\right)\\
 & =\text{Var\ensuremath{\left(\left[h_{0}\left(X_{i}\right)\right]_{+}\right)}}+\text{Var}\left(\nu^{*}\left(X_{i},D_{i}\right)\e_{i}\right)\\
 & =\text{Var\ensuremath{\left(\left[h_{0}\left(X_{i}\right)\right]_{+}\right)}}+ {\s}_{W}^{2}.
\end{align*}
\end{proof}
\begin{proof}[Proof of Theorem \ref{thm:Vfunc_IntF}]
Recall that 
\[
\t_{0}=V\left(h_{0}\right):=\int\ind\left\{ h_{0}\left(x\right)\geq0\right\} v_{0}\left(x\right)f(x)dx,~~~\hat{\t}=V\left(\hat{h}\right)={\ol V}\left(\hat{\mu}\right).
\]
Taking the functional derivative according to Chen \& Gao (2025),
we obtain the following submanifold integral with submanifold dimension
$m=d-1$, 
\begin{align*}
D_h V\left(h_{0}\right)\left[\nu\right] & :=\int_{\left\{x\in\R^d: h_{0}\left(x\right)=0\right\} }\frac{\nu\left(x\right)}{\norm{\Dif_x h_{0}\left(x\right)}}v_{0}\left(x\right)f\left(x\right)d{\cal H}^{d-1}\left(x\right).
\end{align*}
Equivalently, using the notation $\mu_{0}$, we have 
\begin{align*}
D_{\mu}\ol{V}\left(\mu_{0}\right)\left[\nu\right] & :=\int_{\left\{x\in\R^d: h_{0}\left(x\right)=0\right\} }\frac{\left(\nu\left(x,1\right)-\nu\left(x,0\right)\right)}{\norm{\Dif_x h_{0}\left(x\right)}}v_{0}\left(x\right)f\left(x\right)d{\cal H}^{d-1}\left(x\right).
\end{align*}
Applying Theorem 4 and Proposition 2 in \citet{chen2025semiparametric}, we obtain that 
\[
\frac{\sqrt{n}\left(\hat{\t}-\t_{0}\right)}{\s_{V,n}}\dto\cN\left(0,1\right)\quad\text{with }\s_{V,n}^{2} \asymp K_{n}^{\frac{1}{d}}.
\]
\end{proof}
\begin{proof}[Proof of Theorem \ref{thm:Vfunc_Mean}]
We now work with the following rescaled empirical process decomposition
\begin{align}
\sqrt{\frac{n}{\s_{V,n}^{2}}}\left(\P_{n}g_{\mu}-Pg_{\mu_{0}}\right) & =\sqrt{\frac{n}{\s_{V,n}^{2}}}P\left(g_{\hat{\mu}}-g_{\mu_{0}}\right)+\s_{V,n}^{-1}\GG_{n}g_{\mu_{0}}+\s_{V,n}^{-1}\GG_{n}\left(g_{\hat{\mu}}-g_{\mu_{0}}\right),\label{eq:Emp_Decomp-1}
\end{align}
where $g_{\mu}\left(x\right):=\ind\left\{ \mu\left(x,1\right)-\mu\left(x,0\right)\geq0\right\} $.

Note that the term $P\left(g_{\hat{\mu}}-g_{\mu_{0}}\right)\equiv\ol V\left(\mu\right)-\ol V\left(\mu_{0}\right)$
has been analyzed in Section \ref{subsec:ValueFunc_IntF},
where we have shown that 
\[
\sqrt{\frac{n}{\s_{V,n}^{2}}}P\left(g_{\hat{\mu}}-g_{\mu_{0}}\right)\dto\cN\left(0,1\right).
\]
Given the above, the term $\s_{V,n}^{-1}\GG_{n}g_{\mu_{0}}=O_{p}\left(\sqrt{K_{n}^{-1/d}}\right)=o_{p}\left(1\right)$
becomes asymptotically negligible. Below we seek to show that, in
fact, the last term is also asymptotically negligible: 
\[
\s_{V,n}^{-1}\GG_{n}\left(g_{\hat{\mu}}-g_{\mu_{0}}\right)=o_{p}\left(1\right).
\]

Note that $g_{\mu}\left(x\right)$ involves a discontinuous indicator
function, which does not preserve the Donsker property for the Holder
class in general.\footnote{In contrast, the indicator function transformation preserves the Donsker
property for VC classes of functions, in which case the Donsker property
would deliver $\GG_{n}\left(g_{\hat{\mu}}-g_{\mu_{0}}\right)=o_{p}\left(1\right)$,
which implies the weaker condition $K_{n}^{-1/d}\GG_{n}\left(g_{\hat{\mu}}-g_{\mu_{0}}\right)=o_{p}\left(1\right)$
required in this paper.} We thus directly derive the Donsker property via a maximal inequality
on the functional class 
\[
{\cal G}_{a_{n}}:=\left\{ g_{\mu}-g_{\mu_{0}}:\norm{\mu-\mu_{0}}_{\infty}\leq a_{n}\right\} .
\]

We first obtain an envelope function for ${\cal G}_{a_{n}}$ and its
second moment: 
\begin{align*}
 & \left|g_{\mu}\left(x\right)-g_{\mu_{0}}\left(x\right)\right|\\
=\  & \left|\ind\left\{ \mu\left(x,1\right)-\mu\left(x,0\right)\geq0\right\} -\ind\left\{ \mu_{0}\left(x,1\right)-\mu_{0}\left(x,0\right)\geq0\right\} \right|\\
=\  & \ind\left\{ \mu\left(x,1\right)-\mu\left(x,0\right)\geq0>\mu_{0}\left(x,1\right)-\mu_{0}\left(x,0\right)\right\} \\
 & +\ind\left\{ \mu_{0}\left(x,1\right)-\mu_{0}\left(x,0\right)\geq0>\mu\left(x,1\right)-\mu\left(x,0\right)\right\} \\
\leq\  & \ind\left\{ \mu_{0}\left(x,1\right)-\mu_{0}\left(x,0\right)+2\norm{\mu-\mu_{0}}_{\infty}\geq0>\mu_{0}\left(x,1\right)-\mu_{0}\left(x,0\right)\right\} \\
 & +\ind\left\{ \mu_{0}\left(x,1\right)-\mu_{0}\left(x,0\right)\geq0>\mu_{0}\left(x,1\right)-\mu_{0}\left(x,0\right)-2\norm{\mu-\mu_{0}}_{\infty}\right\} \\
\leq\  & \ind\left\{ \mu_{0}\left(x,1\right)-\mu_{0}\left(x,0\right)+2a_{n}\geq0>\mu_{0}\left(x,1\right)-\mu_{0}\left(x,0\right)\right\} \\
 & +\ind\left\{ \mu_{0}\left(x,1\right)-\mu_{0}\left(x,0\right)\geq0>\mu_{0}\left(x,1\right)-\mu_{0}\left(x,0\right)-2a_{n}\right\} \\
=\  & \ind\left\{ \left|\mu_{0}\left(x,1\right)-\mu_{0}\left(x,0\right)\right|\leq2a_{n}\right\} \\
=:\  & G_{a_{n}}
\end{align*}
with 
\begin{align*}
PG_{a_{n}}^{2} & =\P\left(\left|\mu_{0}\left(X_{i},1\right)-\mu_{0}\left(X_{i},0\right)\right|\leq2a_{n}\right)\\
 & =\P\left(\left|h_{0}\left(X\right)\right|\leq2a_{n}\right)\\
 & \leq Ma_{n}.
\end{align*}
Then, provided a finite uniform entropy integral $J_{{\cal G}}$,
\[
P\sup_{\norm{\mu-\mu_{0}}_{\infty}\leq a_{n}}\left|\GG_{n}\left(g_{\mu}-g_{\mu_{0}}\right)\right|\leq J_{{\cal G}}\sqrt{PG_{a_{n}}^{2}}\leq M\sqrt{a_{n}}
\]
and thus 
\[
\s_{V,n}^{-1}\GG_{n}\left(g_{\mu}-g_{\mu_{0}}\right)=O_{p}\left(\sqrt{K_{n}^{-1/d}a_{n}}\right)=o_{p}\left(1\right).
\]

Hence, 
\begin{align*}
\sqrt{\frac{n}{\s_{V,n}^{2}}}\left(\hat{V}\left(\hat{h}\right)-V\left(h_{0}\right)\right) & \equiv\sqrt{\frac{n}{\s_{V,n}^{2}}}\left(\hat{\ol V}\left(\hat{\mu}\right)-\ol V\left(\mu_{0}\right)\right)\\
 & =\sqrt{\frac{n}{\s_{V,n}^{2}}}\left(\ol V\left(\hat{\mu}\right)-\ol V\left(\mu_{0}\right)\right)+o_{p}\left(1\right)\\
 & \dto\cN\left(0,1\right).
\end{align*}
\end{proof}

\section{Additional Simulation Results}
\subsection{Theorem 1} \label{tab:Thm1_Sim_GAM}
In principle, a variety of nonparametric estimators can be employed in the first step for the nuisance parameters, provided their convergence rates are sufficiently fast. In this subsection, we assess the estimation and inference performance of the welfare functional estimator under a known target distribution $F$, using a generalized additive model (GAM) as the first-stage estimator. Specifically, for Models 1-3, we use B-splines as the smooth terms in GAM, while for Models 4-7, we adopt the default tensor product smooths with cubic regression splines.

The simulation results are presented in Table~\ref{tab:Thm1_GAM}. In Models 6 and 7, the confidence intervals exhibit overcoverage relative to the nominal 95\% level. This likely reflects the slow convergence of the variance plug-in estimator toward the true asymptotic variance of the welfare estimator, potentially due to suboptimal choices of smooth terms in the GAM specification.

\begin{table}[!htbp]
\centering
\caption{Theorem 1 Simulation Results for Models 1--7 (GAM)}
\label{tab:Thm1_GAM}
\begin{tabular}{lccccccc}
\toprule
Model & $n$ & $W_{\text{true}}$ & Bias & SD & SE & SD(SE) & Coverage \\
\midrule
M1 & 1500 & 0.3857 &  0.0062 & 0.0474 & 0.0481 & 0.0035 & 0.9500 \\
        & 3000 & 0.3857 &  0.0025 & 0.0341 & 0.0339 & 0.0018 & 0.9525 \\
        & 6000 & 0.3857 & -0.0009 & 0.0231 & 0.0239 & 0.0009 & 0.9545 \\
\midrule
M2 & 1500 & 0.2358 & -0.0092 & 0.0445 & 0.0513 & 0.0027 & 0.9670 \\
        & 3000 & 0.2358 & -0.0057 & 0.0335 & 0.0365 & 0.0013 & 0.9610 \\
        & 6000 & 0.2358 & -0.0048 & 0.0239 & 0.0259 & 0.0006 & 0.9615 \\
\midrule
M3 & 1500 & 0.5001 &  0.0037 & 0.0587 & 0.0604 & 0.0022 & 0.9580 \\
        & 3000 & 0.5001 &  0.0021 & 0.0409 & 0.0430 & 0.0011 & 0.9605 \\
        & 6000 & 0.5001 & -0.0002 & 0.0294 & 0.0305 & 0.0006 & 0.9600 \\
\midrule
M4 & 1500 & 0.1033 &  0.0303 & 0.0465 & 0.0550 & 0.0065 & 0.9615 \\
        & 3000 & 0.1033 &  0.0193 & 0.0338 & 0.0394 & 0.0036 & 0.9635 \\
        & 6000 & 0.1033 &  0.0130 & 0.0243 & 0.0281 & 0.0021 & 0.9610 \\
\midrule
M5 & 1500 & 0.0499 &  0.0316 & 0.0407 & 0.0527 & 0.0087 & 0.9685 \\
        & 3000 & 0.0499 &  0.0217 & 0.0294 & 0.0372 & 0.0050 & 0.9680 \\
        & 6000 & 0.0499 &  0.0138 & 0.0218 & 0.0261 & 0.0030 & 0.9635 \\
\midrule
M6 & 1500 & 0.2315 &  0.0035 & 0.0450 & 0.0634 & 0.0039 & 0.9930 \\
        & 3000 & 0.2315 &  0.0004 & 0.0314 & 0.0448 & 0.0019 & 0.9965 \\
        & 6000 & 0.2315 & -0.0008 & 0.0234 & 0.0317 & 0.0010 & 0.9920 \\
\midrule
M7 & 1500 & 0.1250 &  0.0040 & 0.0380 & 0.0545 & 0.0064 & 0.9950 \\
        & 3000 & 0.1250 &  0.0026 & 0.0272 & 0.0385 & 0.0030 & 0.9920 \\
        & 6000 & 0.1250 &  0.0016 & 0.0199 & 0.0272 & 0.0015 & 0.9915 \\
\bottomrule
\end{tabular}

\begin{flushleft}
\footnotesize \textit{Notes:} (1) $W_{\text{true}}$ denotes the true welfare functional, computed numerically using $5{,}000$ Sobol points drawn from $F$.  
(2) Bias is the average deviation from $W_{\text{true}}$.  
(3) SD is the sampling standard deviation.  
(4) SE is the average estimated asymptotic standard error of $\hat{\overline{W}}(\hat{\mu})$.  
(5) SD(SE) is the standard deviation of SE across iterations.  
(6) Coverage is the empirical coverage probability of the 95\% CI.  
\end{flushleft}
\end{table}

\subsection{Theorem 2} \label{tab:Thm2_Sim_GAM}
We next examine the estimation and inference performance of the welfare functional estimator under an unknown target distribution $F$, using GAM as the first-stage estimator. The GAM specifications are identical to those described in the previous subsection. The simulation results, reported in Table~\ref{tab:Thm2_GAM}, demonstrate that our proposed inference procedure exhibits strong finite-sample performance.

\begin{table}[!htbp]
\centering
\caption{Theorem 2 Simulation Results for Models 1--7 (GAM)}
\label{tab:Thm2_GAM}
\small
\begin{tabular}{lrrrrrrr}
\toprule
Model & $n$ & $W_{\text{true}}$ & Bias & SD & SE & SD(SE) & Coverage \\
\midrule
M8 & 1,500 & 0.3857 & 0.0077 & 0.0412 & 0.0421 & 0.0025 & 0.951 \\
        & 3,000 & 0.3857 & 0.0036 & 0.0296 & 0.0297 & 0.0013 & 0.952 \\
        & 6,000 & 0.3857 & 0.0008 & 0.0206 & 0.0209 & 0.0007 & 0.956 \\
\midrule
M9 & 1,500 & 0.2358 & 0.0000 & 0.0414 & 0.0436 & 0.0021 & 0.965 \\
        & 3,000 & 0.2358 & 0.0008 & 0.0297 & 0.0310 & 0.0010 & 0.962 \\
        & 6,000 & 0.2358 & 0.0002 & 0.0206 & 0.0220 & 0.0005 & 0.960 \\
\midrule
M10 & 1,500 & 0.5001 & 0.0059 & 0.0512 & 0.0514 & 0.0018 & 0.952 \\
        & 3,000 & 0.5001 & 0.0033 & 0.0357 & 0.0366 & 0.0010 & 0.957 \\
        & 6,000 & 0.5001 & 0.0010 & 0.0252 & 0.0260 & 0.0005 & 0.962 \\
\midrule
M11 & 1,500 & 0.1033 & 0.0266 & 0.0362 & 0.0396 & 0.0036 & 0.938 \\
        & 3,000 & 0.1033 & 0.0154 & 0.0261 & 0.0284 & 0.0020 & 0.940 \\
        & 6,000 & 0.1033 & 0.0094 & 0.0190 & 0.0202 & 0.0011 & 0.942 \\
\midrule
M12 & 1,500 & 0.0499 & 0.0277 & 0.0324 & 0.0370 & 0.0052 & 0.941 \\
        & 3,000 & 0.0499 & 0.0172 & 0.0232 & 0.0263 & 0.0029 & 0.946 \\
        & 6,000 & 0.0499 & 0.0106 & 0.0169 & 0.0187 & 0.0016 & 0.941 \\
\midrule
M13 & 1,500 & 0.2315 & 0.0062 & 0.0431 & 0.0456 & 0.0024 & 0.964 \\
        & 3,000 & 0.2315 & 0.0027 & 0.0308 & 0.0323 & 0.0012 & 0.959 \\
        & 6,000 & 0.2315 & 0.0007 & 0.0218 & 0.0229 & 0.0006 & 0.957 \\
\midrule
M14 & 1,500 & 0.1250 & 0.0056 & 0.0367 & 0.0394 & 0.0043 & 0.957 \\
        & 3,000 & 0.1250 & 0.0023 & 0.0249 & 0.0276 & 0.0019 & 0.968 \\
        & 6,000 & 0.1250 & 0.0007 & 0.0176 & 0.0195 & 0.0009 & 0.969 \\
\bottomrule
\end{tabular}

\begin{flushleft}
\footnotesize \textit{Notes:} (1) $W_{\text{true}}$ denotes the true welfare functional, computed numerically using $5{,}000$ Sobol points drawn from $F$.  
(2) Bias is the average deviation from $W_{\text{true}}$.  
(3) SD is the sampling standard deviation.  
(4) SE is the average estimated asymptotic standard error of $\hat{\overline{W}}(\hat{\mu})$.  
(5) SD(SE) is the standard deviation of SE across iterations.  
(6) Coverage is the empirical coverage probability of the 95\% CI.  
\end{flushleft}
\end{table}

\subsection{Theorem 3} \label{tab:Thm3_robust}

In addition to the model specification considered in the main text for the Theorem \ref{thm:Vfunc_IntF} simulations, we conducted two auxiliary simulation studies with slightly modified designs: one with the target population $F \sim \text{U}([-1.75,1.75]^2)$ and the other with $F \sim \text{U}([-1.9,1.9]^2)$, keeping all other settings fixed. The results are reported in Table \ref{tab:Thm3_Sim_Support_Sensitivity}. We find that the coverage rate approaches $95\%$ as the sample size increases in both cases. However, in the latter case the supports of the treated and control groups do not fully nest the target population $F$, potentially leading to extrapolation bias of the first-stage estimator and a slight undercoverage.

\begin{table}[!htbp]
\centering
\caption{Sensitivity Checks for Model 15}
\label{tab:Thm3_Sim_Support_Sensitivity}
\begin{tabular}{lcccccc}
\toprule
Case & $n$ & $W_{\text{true}}$ & Bias & SD & SE & Coverage \\
\midrule

(a) $F \sim \text{U}[-1.75,1.75]^2$ & 1500 & 3.1416 & 0.0222 & 0.1289 & 0.1198 & 0.926 \\
                                    & 3000 & 3.1416 & 0.0077 & 0.0801 & 0.0796 & 0.942 \\
                                    & 6000 & 3.1416 & 0.0089 & 0.0542 & 0.0558 & 0.952 \\
\midrule
(b) $F \sim \text{U}[-1.9,1.9]^2$  & 1500 & 3.1416 & 0.0118 & 0.0746 & 0.0676 & 0.923 \\
                                    & 3000 & 3.1416 & 0.0047 & 0.0474 & 0.0453 & 0.938 \\
                                    & 6000 & 3.1416 & 0.0037 & 0.0334 & 0.0317 & 0.936 \\
\bottomrule
\end{tabular}
\end{table}

Besides the sieve dimensions for estimating $\mu_0(x,1)$ and $\mu_0(x,0)$, $\epsilon$ is an additional tuning parameter that determines how well the level set is approximated. To assess the sensitivity of our method to this parameter, we consider the alternative values $\epsilon = 0.0025$ and $\epsilon = 0.0075$ while keeping anything else fixed. Table~\ref{tab:Thm3_Sim_Epsilon_Sensitivity} shows that the simulation results are not materially affected by the choice of $\epsilon$.

\begin{table}[!htbp]
\centering
\caption{Sensitivity Checks for Model 15}
\label{tab:Thm3_Sim_Epsilon_Sensitivity}
\begin{tabular}{lcccccc}
\toprule
Case & $n$ & $W_{\text{true}}$ & Bias & SD & SE & Coverage \\
\midrule
(a) $\epsilon = 0.0075$ & 1500 & 3.1416 & 0.0120 & 0.0642 & 0.0635 & 0.9395 \\
                        & 3000 & 3.1416 & 0.0126 & 0.0442 & 0.0448 & 0.9445 \\
                        & 6000 & 3.1416 & 0.0102 & 0.0311 & 0.0316 & 0.9445 \\
\midrule
(b) $\epsilon = 0.0025$ & 1500 & 3.1416 & 0.0120 & 0.0642 & 0.0638 & 0.9420 \\
                        & 3000 & 3.1416 & 0.0126 & 0.0442 & 0.0450 & 0.9445 \\
                        & 6000 & 3.1416 & 0.0102 & 0.0311 & 0.0317 & 0.9450 \\
\bottomrule
\end{tabular}
\end{table}

\subsection{Simulation Results for DML}
An alternative to the semi-parametric two-step estimation of the welfare functional is to use the double-debiased machine learning approach. The DML estimator is implemented via a cross-fitting scheme. Specifically, the sample is partitioned into K folds. For each fold, we use the observations in the remaining $K-1$ folds to obtain nonparametric estimates of the nuisance functions $\mu_0(x,1)$ and $\mu_0(x,0)$ (in our case, using GAM). These estimates are then applied to the held-out fold to take an average of the fitted indicator \(\mathbbm{1}\{\hat{\mu}_0(x,1) - \hat{\mu}_0(x,0) \geq 0\}\). This process is repeated so that each fold serves once as the testing set, and the final estimator is obtained by averaging the resulting values across all folds. A confidence interval could then be constructed as in Theorem \ref{thm:Welfare_mean} with the debiased ML estimate in replacement of the semi-parametric two-step estimate.

We simulate 1,000 iterations for Models 1 to 3 with $K = 5$, and the results are summarized in table \ref{Thm3}. The confidence intervals exhibit slight undercoverage in small samples, but this issue diminishes as the sample size increases. This pattern is consistent with the asymptotic nature of the interval's construction.

\begin{table}[!htbp]
\centering
\caption{Theorem DML Simulation Results}
\label{Thm3}
\small
\begin{tabular}{lrrrrrr}
\toprule
\textbf{Model} & \textbf{n} & \textbf{$W_0$} & \textbf{Bias} & \textbf{SD} & \textbf{SE} & \textbf{Coverage} \\
\midrule
Model 1 & 1,500 & 0.3857 & -0.01271 & 0.04583 & 0.04215 & 0.906 \\
Model 1 & 3,000 & 0.3857 & -0.00776 & 0.03184 & 0.02969 & 0.921 \\
Model 1 & 6,000 & 0.3857 & -0.00392 & 0.02193 & 0.02092 & 0.934 \\
\addlinespace
Model 2 & 1,500 & 0.2358 & -0.01051 & 0.04900 & 0.04372 & 0.916 \\
Model 2 & 3,000 & 0.2358 & -0.00475 & 0.03319 & 0.03097 & 0.934 \\
Model 2 & 6,000 & 0.2358 & -0.00195 & 0.02297 & 0.02200 & 0.929 \\
\addlinespace
Model 3 & 1,500 & 0.5001 & -0.00568 & 0.05636 & 0.05134 & 0.913 \\
Model 3 & 3,000 & 0.5001 & -0.00261 & 0.03701 & 0.03654 & 0.953 \\
Model 3 & 6,000 & 0.5001 & -0.00174 & 0.02628 & 0.02597 & 0.954 \\
\bottomrule
\end{tabular}
\end{table}

\subsection{Sieve Variance Estimation} \label{SieveVar}

\begin{table}[!htbp]
\centering
\caption{Theorem 1 Simulation Results for Models 1--3 (Sieve Variance)}
\label{tab:Sim_Models123}
\begin{tabular}{lccccccc}
\toprule
Model & $n$ & $W_{\text{true}}$ & Bias & SD & SE & SD(SE) & Coverage \\
\midrule
M1 & 1500 & 0.3857 & 0.0152 & 0.0469 & 0.0467 & 0.0035 & 0.9370 \\
   & 3000 & 0.3857 & 0.0076 & 0.0337 & 0.0329 & 0.0019 & 0.9390 \\
   & 6000 & 0.3857 & 0.0036 & 0.0229 & 0.0232 & 0.0011 & 0.9510 \\
\midrule
M2 & 1500 & 0.2358 & 0.0039 & 0.0481 & 0.0488 & 0.0029 & 0.9480 \\
   & 3000 & 0.2358 & 0.0042 & 0.0350 & 0.0345 & 0.0014 & 0.9425 \\
   & 6000 & 0.2358 & 0.0025 & 0.0241 & 0.0245 & 0.0007 & 0.9535 \\
\midrule
M3 & 1500 & 0.5001 & 0.0038 & 0.0580 & 0.0581 & 0.0023 & 0.9560 \\
   & 3000 & 0.5001 & 0.0019 & 0.0407 & 0.0413 & 0.0012 & 0.9560 \\
   & 6000 & 0.5001 & -0.0002 & 0.0293 & 0.0294 & 0.0006 & 0.9525 \\
\bottomrule
\end{tabular}

\begin{flushleft}
\footnotesize \textit{Notes:} (1) $W_{\text{true}}$ denotes the true welfare functional, computed numerically using $5{,}000$ Sobol points drawn from $F$.  
(2) Bias is the average deviation from $W_{\text{true}}$.  
(3) SD is the sampling standard deviation.  
(4) SE is the average estimated asymptotic standard error of $\hat{\overline{W}}(\hat{\mu})$.  
(5) SD(SE) is the standard deviation of SE across iterations.  
(6) Coverage is the empirical coverage probability of the 95\% CI.  
\end{flushleft}
\end{table}

\begin{table}[!htbp]
\centering
\caption{Theorem 1 Simulation Results for Models 4--7 (Sieve Variance)}
\label{tab:Sim_Models4to7}
\begin{tabular}{lccccccc}
\toprule
Model & $n$ & $W_{\text{true}}$ & Bias & SD & SE & SD(SE) & Coverage \\
\midrule
M4 & 1500 & 0.1033 & 0.0185 & 0.0478 & 0.0482 & 0.0088 & 0.9400 \\
   & 3000 & 0.1033 & 0.0088 & 0.0348 & 0.0350 & 0.0048 & 0.9445 \\
   & 6000 & 0.1033 & 0.0043 & 0.0248 & 0.0251 & 0.0024 & 0.9490 \\
\midrule
M5 & 1500 & 0.0499 & 0.0282 & 0.0418 & 0.0431 & 0.0110 & 0.9310 \\
   & 3000 & 0.0499 & 0.0158 & 0.0303 & 0.0311 & 0.0068 & 0.9335 \\
   & 6000 & 0.0499 & 0.0094 & 0.0227 & 0.0224 & 0.0038 & 0.9290 \\
\midrule
M6 & 1500 & 0.2315 & 0.0268 & 0.0559 & 0.0552 & 0.0057 & 0.9260 \\
   & 3000 & 0.2315 & 0.0123 & 0.0394 & 0.0402 & 0.0030 & 0.9480 \\
   & 6000 & 0.2315 & 0.0050 & 0.0288 & 0.0288 & 0.0015 & 0.9490 \\
\midrule
M7 & 1500 & 0.1250 & 0.0462 & 0.0505 & 0.0502 & 0.0087 & 0.8920 \\
   & 3000 & 0.1250 & 0.0218 & 0.0346 & 0.0350 & 0.0046 & 0.9270 \\
   & 6000 & 0.1250 & 0.0101 & 0.0245 & 0.0245 & 0.0023 & 0.9400 \\
\bottomrule
\end{tabular}

\begin{flushleft}
\footnotesize \textit{Notes:} (1) $W_{\text{true}}$ denotes the true welfare functional, computed numerically using $5{,}000$ Sobol points drawn from $F$.  
(2) Bias is the average deviation from $W_{\text{true}}$.  
(3) SD is the sampling standard deviation.  
(4) SE is the average estimated asymptotic standard error of $\hat{\overline{W}}(\hat{\mu})$.  
(5) SD(SE) is the standard deviation of SE across iterations.  
(6) Coverage is the empirical coverage probability of the 95\% CI.  
\end{flushleft}

\end{table}

\begin{table}[!htbp]
\centering
\caption{Theorem 2 Simulation Results for Models 8--10 (Sieve Variance)}
\label{tab:Sim_Models1to3}
\begin{tabular}{lccccccc}
\toprule
Model & $n$ & $W_{\text{true}}$ & Bias & SD & SE & SD(SE) & Coverage \\
\midrule
M8 & 1500 & 0.3857 & 0.0068 & 0.0414 & 0.0417 & 0.0026 & 0.9475 \\
   & 3000 & 0.3857 & 0.0029 & 0.0296 & 0.0294 & 0.0014 & 0.9505 \\
   & 6000 & 0.3857 & 0.0006 & 0.0206 & 0.0208 & 0.0007 & 0.9555 \\
\midrule
M9 & 1500 & 0.2358 & 0.0042 & 0.0425 & 0.0431 & 0.0026 & 0.9555 \\
   & 3000 & 0.2358 & 0.0029 & 0.0304 & 0.0305 & 0.0012 & 0.9475 \\
   & 6000 & 0.2358 & 0.0012 & 0.0209 & 0.0216 & 0.0006 & 0.9510 \\
\midrule
M10 & 1500 & 0.5001 & 0.0067 & 0.0511 & 0.0511 & 0.0020 & 0.9480 \\
   & 3000 & 0.5001 & 0.0037 & 0.0357 & 0.0364 & 0.0011 & 0.9590 \\
   & 6000 & 0.5001 & 0.0013 & 0.0253 & 0.0259 & 0.0006 & 0.9580 \\
\bottomrule
\end{tabular}

\begin{flushleft}
\footnotesize \textit{Notes:} (1) $W_{\text{true}}$ denotes the true welfare functional, computed numerically using $5{,}000$ Sobol points drawn from $F$.  
(2) Bias is the average deviation from $W_{\text{true}}$.  
(3) SD is the sampling standard deviation.  
(4) SE is the average estimated asymptotic standard error of $\hat{\overline{W}}(\hat{\mu})$.  
(5) SD(SE) is the standard deviation of SE across iterations.  
(6) Coverage is the empirical coverage probability of the 95\% CI.  
\end{flushleft}
\end{table}

\begin{table}[!htbp]
\centering
\caption{Theorem 2 Simulation Results for Models 11--14 (Sieve Variance)}
\label{tab:Sim_Models4to7_New}
\begin{tabular}{lccccccc}
\toprule
Model & $n$ & $W_{\text{true}}$ & Bias & SD & SE & SD(SE) & Coverage \\
\midrule
M11 & 1500 & 0.1033 & 0.0307 & 0.0365 & 0.0379 & 0.0041 & 0.9065 \\
    & 3000 & 0.1033 & 0.0156 & 0.0264 & 0.0272 & 0.0023 & 0.9290 \\
    & 6000 & 0.1033 & 0.0084 & 0.0192 & 0.0195 & 0.0013 & 0.9395 \\
\midrule
M12 & 1500 & 0.0499 & 0.0418 & 0.0316 & 0.0344 & 0.0050 & 0.8410 \\
    & 3000 & 0.0499 & 0.0232 & 0.0229 & 0.0245 & 0.0031 & 0.8930 \\
    & 6000 & 0.0499 & 0.0126 & 0.0168 & 0.0175 & 0.0018 & 0.9205 \\
\midrule
M13 & 1500 & 0.2315 & 0.0177 & 0.0432 & 0.0438 & 0.0030 & 0.9420 \\
    & 3000 & 0.2315 & 0.0088 & 0.0309 & 0.0313 & 0.0016 & 0.9440 \\
    & 6000 & 0.2315 & 0.0041 & 0.0221 & 0.0222 & 0.0008 & 0.9495 \\
\midrule
M14 & 1500 & 0.1250 & 0.0251 & 0.0382 & 0.0386 & 0.0050 & 0.9230 \\
    & 3000 & 0.1250 & 0.0117 & 0.0258 & 0.0268 & 0.0024 & 0.9440 \\
    & 6000 & 0.1250 & 0.0054 & 0.0182 & 0.0188 & 0.0012 & 0.9415 \\
\bottomrule
\end{tabular}

\begin{flushleft}
\footnotesize \textit{Notes:} (1) $W_{\text{true}}$ denotes the true welfare functional, computed numerically using $5{,}000$ Sobol points drawn from $F$.  
(2) Bias is the average deviation from $W_{\text{true}}$.  
(3) SD is the sampling standard deviation.  
(4) SE is the average estimated asymptotic standard error of $\hat{\overline{W}}(\hat{\mu})$.  
(5) SD(SE) is the standard deviation of SE across iterations.  
(6) Coverage is the empirical coverage probability of the 95\% CI.  
\end{flushleft}
\end{table}

\begin{table}[htbp]
\centering
\caption{Estimated Welfare Gains and Share of Population to be Treated Under Nonparametric Plug-in Rule (Sieve Variance)}
\label{tab:emp_compare}
\begin{subtable}{\textwidth}
\centering
\caption{30-Month Post-Program Earnings, No Treatment Cost}
\begin{tabular}{@{}p{6cm}cc@{}}
\toprule
\textbf{Method} & \textbf{Share Treated} & \textbf{Est. Welfare Gain} \\
\midrule
Ours         & \makecell{0.89\\ \scriptsize{(0.73 1.05)}} & \makecell{\$1,519  \\ \scriptsize{(\$691, \$2347)}}
\\
\midrule
KT (2018) & \makecell{0.91 \\ \scriptsize{NA}} & \makecell{\$1,693  \\ \scriptsize{NA}} \\
\bottomrule
\end{tabular}
\end{subtable}

\vspace{1em}

\begin{subtable}{\textwidth}
\centering
\caption{30-Month Post-Program Earnings, \$774 Cost per Treatment}
\begin{tabular}{@{}p{6cm}cc@{}}
\toprule
\textbf{Method} & \textbf{Share Treated} & \textbf{Est. Welfare Gain} \\
\midrule
Ours         & \makecell{0.80\\ \scriptsize{(0.53, 1.07)}} & \makecell{\$858  \\ \scriptsize{(\$84, \$1631)}}
\\
\midrule
KT (2018) & \makecell{0.78 \\ \scriptsize{NA}} & \makecell{\$996 \\ \scriptsize{NA}} \\
\bottomrule
\end{tabular}
\end{subtable}

\vspace{1em}
\begin{flushleft}
\small \textit{Note:} Two-sided 95\% confidence intervals in parentheses, constructed based on the asymptotic distributions of the corresponding estimators.
\end{flushleft}
\end{table}

\section{Additional Results for Empirical Application} \label{Appendix_Emp}

\subsection{Sensitivity Analysis for Tuning Parameters} \label{Appendix_Emp_Tuning}

When conducting inference for the value functional in the empirical analysis of the JTPA data, two tuning parameters must be specified: $\iota$ and $s$. The parameter $\iota$ determines the size of the set $\{x \in \chi : -\epsilon < \hat{h}(x) < \epsilon\}$. The parameter $s$ controls the smoothness of the estimated density function over the trimmed dataset by scaling the bandwidths used in kernel density estimation.

To evaluate robustness, we perform sensitivity analyses with respect to both tuning parameters. The results, reported in Table~\ref{tab:emp_iota_s_sensitivity}, show that the estimated value functional $\hat{\overline{V}}(\hat{\mu})$ is largely insensitive to the choice of $\iota$, but more responsive to the choice of $s$. This shows the importance of choosing an appropriate kernel density estimator.

\begin{table}[!htbp]
\centering
\caption{Sensitivity of $\hat{\overline{V}}(\hat{\mu})$ to Tuning Parameters $\iota$ and $s$ (Cost = F vs.\ Cost = T)}
\label{tab:emp_iota_s_sensitivity}
\begin{subtable}{\textwidth}
\centering
\caption{Sensitivity to $\iota$}
\begin{tabular}{ccccccc}
\toprule
 & $\iota$ & $\hat{V}$ & SE & CI Low & CI High & Num \\
\midrule
\multicolumn{7}{l}{\textbf{Cost = F}} \\
\midrule
 & 0.005  & 0.8908 & 0.0866 & 0.7210 & 1.0606 & 475 \\
 & 0.0075 & 0.8908 & 0.0862 & 0.7219 & 1.0597 & 713 \\
 & 0.0100 & 0.8908 & 0.0828 & 0.7286 & 1.0530 & 931 \\
 & 0.0125 & 0.8908 & 0.0840 & 0.7262 & 1.0554 & 1203 \\
 & 0.0150 & 0.8908 & 0.0857 & 0.7228 & 1.0588 & 1463 \\
\midrule
\multicolumn{7}{l}{\textbf{Cost = T}} \\
\midrule
 & 0.005  & 0.7971 & 0.1343 & 0.5338 & 1.0603 & 689 \\
 & 0.0075 & 0.7971 & 0.1303 & 0.5417 & 1.0525 & 1044 \\
 & 0.0100 & 0.7971 & 0.1373 & 0.5279 & 1.0662 & 1411 \\
 & 0.0125 & 0.7971 & 0.1399 & 0.5229 & 1.0712 & 1784 \\
 & 0.0150 & 0.7971 & 0.1406 & 0.5214 & 1.0727 & 2152 \\
\bottomrule
\end{tabular}
\end{subtable}

\vspace{1em}

\begin{subtable}{\textwidth}
\centering
\caption{Sensitivity to $s$}
\begin{tabular}{ccccccc}
\toprule
 & $s$ & $\hat{V}$ & SE & CI Low & CI High & Num \\
\midrule
\multicolumn{7}{l}{\textbf{Cost = F}} \\
\midrule
 & 1 & 0.8908 & 0.0896 & 0.7151 & 1.0665 & 931 \\
 & 2 & 0.8908 & 0.0860 & 0.7224 & 1.0593 & 931 \\
 & 3 & 0.8908 & 0.0828 & 0.7286 & 1.0530 & 931 \\
 & 4 & 0.8908 & 0.0809 & 0.7322 & 1.0494 & 931 \\
 & 5 & 0.8908 & 0.0793 & 0.7354 & 1.0462 & 931 \\
\midrule
\multicolumn{7}{l}{\textbf{Cost = T}} \\
\midrule
 & 1 & 0.7971 & 0.1103 & 0.5809 & 1.0132 & 1411 \\
 & 2 & 0.7971 & 0.1254 & 0.5513 & 1.0428 & 1411 \\
 & 3 & 0.7971 & 0.1373 & 0.5279 & 1.0662 & 1411 \\
 & 4 & 0.7971 & 0.1420 & 0.5187 & 1.0754 & 1411 \\
 & 5 & 0.7971 & 0.1554 & 0.4925 & 1.1016 & 1411 \\
\bottomrule
\end{tabular}
\end{subtable}

\begin{flushleft}
\small \textit{Notes:} (a) Cost = T if the outcome variable is the 30-month earnings minus an additional \$774 and F otherwise.  
(b) SE = $\hat{\sigma}_V/\sqrt{n}$.  
(c) $\epsilon = \iota \times \text{SD}(\hat{h}(x))$ over \textbf{X\_trim}.  
(d) $\text{CI Low} = \hat{\overline{V}}(\hat{\mu}) - 1.96 \times \text{SE}$, $\text{CI High} = \hat{\overline{V}}(\hat{\mu}) + 1.96 \times \text{SE}$.  
(e) Num is the number of Sobol points whose $\hat{h}$ evaluations fall into $\{x \in \textbf{X\_trim} : |\hat{h}(x)| < \epsilon\}$.  
\end{flushleft}
\end{table}

\subsection{Sensitivity Analysis for Density Estimation} \label{Appendix_Emp_Density}

We employed a naive Gaussian kernel density estimator for the covariate distribution in the trimmed dataset, treating years of education as a continuous variable. While years of education  is theoretically continuous, in practice it takes on only discrete values in the dataset. This motivates considering alternative density estimators that explicitly treat years of education as categorical in order to assess the robustness of our inference. To this end, we partitioned the trimmed dataset by educational level and examined three cases:
(1) a naive Gaussian kernel density estimator with bandwidths selected by \texttt{Hscv()} and scaled by a small factor to ensure smoothness;
(2) a logspline density estimator (\cite*{StoneEtAl1997});
(3) a penalized B-spline density estimator (\cite*{SchellhaseKauermann2012}).

The results are presented in Table~\ref{tab:emp_density_sensitivity}. As shown, when costs are taken into account, the Gaussian kernel density estimator discussed in the main text yields relatively conservative confidence intervals compared to the three alternatives. In contrast, when costs are not considered, the results remain largely unchanged across all the estimators. However, we should emphasize that these three alternative density estimators are applied based on the assumption that years of education is categorical, but in theory it is treated as a continuous variable, which provides justification for using the Gaussian kernel estimator in the main text.

\begin{table}[!htbp]
\centering
\caption{Sensitivity of $\hat{\overline{V}}(\hat{\mu})$ to Alternative Density Estimators}
\label{tab:emp_density_sensitivity}
\small
\begin{tabular}{lcccccccc}
\toprule
 & \multicolumn{4}{c}{\textbf{Cost = T}} & \multicolumn{4}{c}{\textbf{Cost = F}} \\
\cmidrule(lr){2-5} \cmidrule(lr){6-9}
Estimator & $\hat{V}$ & SE & CI Low & CI High & $\hat{V}$ & SE & CI Low & CI High \\
\midrule
Logspline & 0.7971 & 0.0932 & 0.6144 & 0.9797 & 0.8908 & 0.1041 & 0.6868 & 1.0948 \\
Kernel    & 0.7971 & 0.0704 & 0.6590 & 0.9351 & 0.8908 & 0.0980 & 0.6988 & 1.0829 \\
Penalized B-spline & 0.7971 & 0.0601 & 0.6793 & 0.9148 & 0.8908 & 0.0691 & 0.7555 & 1.0262 \\
\bottomrule
\end{tabular}
\begin{flushleft}
\footnotesize \textit{Notes:} (a) Cost = T if the outcome variable is 30-month earnings minus an additional \$774, and Cost = F otherwise.  
(b) SE = $\hat{\sigma}_V/\sqrt{n}$.  
(c) $\text{CI Low} = \hat{\overline{V}}(\hat{\mu}) - 1.96 \times \text{SE}$, $\text{CI High} = \hat{\overline{V}}(\hat{\mu}) + 1.96 \times \text{SE}$.
\end{flushleft}
\end{table}

\subsection{Extrapolating Sieve Estimates}
\label{Appendix_Emp_Extrapolate}

We also examine the case where the dataset is not trimmed. In this setting, interpolation is required to evaluate the estimated functions $\hat{\mu}(x,0)$ and $\hat{\mu}(x,1)$ outside the supports of the control and treated groups, respectively. The formulas for the welfare and value functional estimators, as well as their asymptotic variances, remain unchanged. The only difference is that the sample averages are now computed over the full covariate support rather than the trimmed dataset. The results for this specification are reported in Table~\ref{tab:emp_extrapolate}.

\begin{table}[!htbp]
\centering
\caption{Estimated Welfare Gains and Share of Population to be Treated Under Nonparametric Plug-in Rule Without Trimming}
\label{tab:emp_extrapolate}
\begin{subtable}{\textwidth}
\centering
\caption{30-Month Post-Program Earnings, No Treatment Cost}
\begin{tabular}{@{}p{6cm}cc@{}}
\toprule
\textbf{Method} & \textbf{Share Treated} & \textbf{Est. Welfare Gain} \\
\midrule
Ours         & \makecell{0.92\\ \scriptsize{(0.80 1.03)}} & \makecell{\$1,459  \\ \scriptsize{(\$825, \$2093)}}
\\
\midrule
KT (2018) & \makecell{0.91 \\ \scriptsize{NA}} & \makecell{\$1,693  \\ \scriptsize{NA}} \\
\bottomrule
\end{tabular}
\end{subtable}

\vspace{1em}

\begin{subtable}{\textwidth}
\centering
\caption{30-Month Post-Program Earnings, \$774 Cost per Treatment}
\begin{tabular}{@{}p{6cm}cc@{}}
\toprule
\textbf{Method} & \textbf{Share Treated} & \textbf{Est. Welfare Gain} \\
\midrule
Ours         & \makecell{0.85\\ \scriptsize{(0.71, 0.99)}} & \makecell{\$768  \\ \scriptsize{(\$164, \$1373)}}
\\
\midrule
KT (2018) & \makecell{0.78 \\ \scriptsize{NA}} & \makecell{\$996 \\ \scriptsize{NA}} \\
\bottomrule
\end{tabular}
\end{subtable}

\vspace{1em}
\begin{flushleft}
\small \textit{Note:} Two-sided 95\% confidence intervals in parentheses, constructed based on the asymptotic distributions of the corresponding estimators.
\end{flushleft}
\end{table}

\subsection{Estimating $\hat{\Omega}$ Using the Trimmed Dataset} \label{Appendix_Emp_Omega}

We used the full sample to estimate the asymptotic variance-covariance matrix $\hat{\Omega}$. For comparison, Table~\ref{tab:emp_trim_for_omega} reports the results obtained when $\hat{\Omega}$ is instead estimated using only the trimmed dataset. We observe that the CIs expands by a small factor.

\begin{table}[htbp]
\centering
\caption{Estimated Share of Population to be Treated Under Nonparametric Plug-in Rule With $\hat{\O}$ Estimated Using Trimmed Data}
\label{tab:emp_trim_for_omega}
\begin{subtable}{\textwidth}
\centering
\caption{30-Month Post-Program Earnings, No Treatment Cost}
\begin{tabular}{@{}p{6cm}c@{}}
\toprule
\textbf{Method} & \textbf{Share Treated} \\
\midrule
Ours         & \makecell{0.89\\ \scriptsize{(0.72, 1.06)}} \\
\midrule
KT (2018) & \makecell{0.91 \\ \scriptsize{NA}} \\
\bottomrule
\end{tabular}
\end{subtable}

\vspace{1em}

\begin{subtable}{\textwidth}
\centering
\caption{30-Month Post-Program Earnings, \$774 Cost per Treatment}
\begin{tabular}{@{}p{6cm}c@{}}
\toprule
\textbf{Method} & \textbf{Share Treated} \\
\midrule
Ours         & \makecell{0.80\\ \scriptsize{(0.44, 1.15)}} \\
\midrule
KT (2018) & \makecell{0.78 \\ \scriptsize{NA}} \\
\bottomrule
\end{tabular}
\end{subtable}
\end{table}

\subsection{The Sieve Score Bootstrapping Procedure for Critical Values}
\label{Appendix_Emp_CriticalVal}

Following \cite*{chen2018optimal}, one could also use the sieve score bootstrap procedure to calculate the critical value used in the construction of the 95\% CI for $\overline{V}(\mu_0)$. The algorithm starts with making iid draws $\{\omega_i\}_{i=1}^n$ from a distribution independent of the data \textbf{df} with mean zero, unit variance, and finite third moment (e.g. a standard normal distribution), and then calculates the the bootstrap sieve t-statistic using the formula
\[
Z_n^* = \frac{DV(\hat{\mu})[\nu]' (B'B/n)}{\hat{\sigma}_V/\sqrt{n}}
\begin{bmatrix}
    \frac{1}{\sqrt{n}}\sum_{i = 1}^{n_1} \psi^{K_1}(x_i) \hat{u}_i \omega_i \\

    \frac{1}{\sqrt{n}}\sum_{j = n_1 + 1}^n \psi^{K_0}(x_j) \hat{u}_j \omega_j\textbf{}        
\end{bmatrix}
\]
with $B$ being the design matrix of regression $Y_i$ on $D_i \psi^{K_1}(X_i)$ and $(1-D_i )\psi^{K_0}(X_i)$.  The 95\% quantile of the bootstrapped $|Z_n^*|$ could be used as the critical value in the construction of the CI. Setting the number of bootstrap equals 1000, the resulting critical values is 1.865758, not significantly different from 1.96.

\end{document}